\documentclass[twoside]{article}

\usepackage[accepted]{aistats2024}
\usepackage[T1]{fontenc}
\usepackage{amsthm}
\usepackage{subcaption}
\usepackage{pifont}
\usepackage{tabularx,threeparttable}
\usepackage{etoc}
\etocframedstyle[1]{\textbf{\textsc{Table of Contents}}}
\etocsettocdepth{3}
\usepackage{dsfont}

\usepackage[round]{natbib}

\usepackage[T1]{fontenc}    % use 8-bit T1 fonts

\usepackage{booktabs}
\usepackage{nicefrac}       % compact symbols for 1/2, etc.
\usepackage{xargs}

\allowdisplaybreaks
\usepackage{amssymb,mathrsfs}
\usepackage{mathtools}
\usepackage[dvipsnames]{xcolor}
\usepackage{bbm}
\usepackage[colorlinks = true, citecolor = black]{hyperref}
\usepackage{enumitem}

\usepackage{tikz}
\usepackage{aliascnt}
\usepackage{cleveref}
\let\etoolboxforlistloop\forlistloop % save the good meaning of \forlistloop
\usepackage{autonum}
\let\forlistloop\etoolboxforlistloop % restore the good meaning of \forlistloop

\usepackage{accents}

\definecolor{lightred}{rgb}{1, 0.8, 0.8}

\DeclareMathAlphabet{\mathpzc}{OT1}{pzc}{m}{it}

\makeatletter
\newtheorem{theorem}{Theorem}
\crefname{theorem}{theorem}{Theorems}
\Crefname{Theorem}{Theorem}{Theorems}

\newtheorem{lemma}{Lemma}
\newaliascnt{corollary}{theorem}

\aliascntresetthe{corollary}
\crefname{corollary}{corollary}{corollaries}
\Crefname{Corollary}{Corollary}{Corollaries}

\newaliascnt{proposition}{theorem}

\aliascntresetthe{proposition}
\crefname{proposition}{proposition}{propositions}
\Crefname{Proposition}{Proposition}{Propositions}

\newaliascnt{definition}{theorem}

\aliascntresetthe{definition}
\crefname{definition}{definition}{definitions}
\Crefname{Definition}{Definition}{Definitions}

\newaliascnt{remark}{theorem}

\aliascntresetthe{remark}
\crefname{remark}{remark}{remarks}
\Crefname{Remark}{Remark}{Remarks}

\crefname{example}{example}{examples}
\Crefname{Example}{Example}{Examples}

\crefname{figure}{figure}{figures}
\Crefname{Figure}{Figure}{Figures}

\newtheorem{assumption}{\textbf{H}\hspace{-3pt}}
\crefformat{assumption}{{\textbf{H}}#2#1#3}
\newtheorem{assumptionsup}{\textbf{A}\hspace{-3pt}}
\crefformat{assumptionsup}{{\textbf{A}}#2#1#3}
\newtheorem{assumptionnona}{\textbf{N}\hspace{-3pt}}
\crefformat{assumptionnona}{{\textbf{N}}#2#1#3}

\newtheorem{assumptionnonsup}{\textbf{M}\hspace{-3pt}}
\crefformat{assumptionnonsup}{{\textbf{M}}#2#1#3}

\def\mss{\mathsf{S}}

\def\mcy{\mathcal{Y}}

\def\mcz{\mathcal{Z}}

\def\Rset{\mathbb{R}}

\newcommand{\abs}[1]{\left\vert #1 \right\vert}

\newcommandx{\psr}[3][3=]{\left\langle#1,#2 \right\rangle_{#3}}
\newcommandx{\normr}[2][2=]{ \left\Vert#1 \right\Vert_{#2}}
\newcommandx{\psrLigne}[3][3=]{\langle#1,#2 \rangle_{#3}}
\newcommandx{\normrLigne}[2][2=]{ \Vert#1 \Vert_{#2}}

\newcommandx{\norm}[2][1=]{\ifthenelse{\equal{#1}{}}{\left\Vert #2 \right\Vert}{\left\Vert #2 \right\Vert^{#1}}}

\newcommand\probaMarkovTilde[2][2=]
{\ifthenelse{\equal{#2}{}}{{\widetilde{\mathbb{P}}_{#1}}}{\widetilde{\mathbb{P}}_{#1}\left[ #2\right]}}

\def\eqsp{\;}

\newcommand\sequence[3][2=,3=]
{\ifthenelse{\equal{#3}{}}{\ensuremath{\{ #1_{#2}\}}}{\ensuremath{\{ #1_{#2}, \eqsp #2 \in #3 \}}}}
\newcommand\sequenceD[3][2=,3=]
{\ifthenelse{\equal{#3}{}}{\ensuremath{\{ #1_{#2}\}}}{\ensuremath{( #1)_{ #2 \in #3} }}}

\newcommand\sequenceDouble[4][3=,4=]
{\ifthenelse{\equal{#3}{}}{\ensuremath{\{ (#1_{#3},#2_{#3}) \}}}{\ensuremath{\{  (#1_{#3},#2_{#3}), \eqsp #3 \in #4 \}}}}

\def\Id{\mathrm{Id}}

\newcommand{\beq}{\begin{equation}}
\newcommand{\eeq}{\end{equation}}

\def\dz{\operatorname{dz}}
\def\argmin{\operatorname{argmin}}

\newcommand*{\dd}{\mathop{}\!\mathrm{d}}
\def\ee{~}

\newcommandx{\gperthmc}[2][1=,2=]{\ifthenelse{\equal{#1}{}}{\Xi}{\ifthenelse{\equal{#2}{}}{\Xi_{h,#1}}{\Xi_{#2,#1}}}}
\newcommandx{\Phiverlet}[2][1=,2=]{\ifthenelse{\equal{#1}{}}{\Phi}{\Phi_{#1}^{\circ (#2)}}}
\newcommandx{\gpertub}[2][1=,2=]{\ifthenelse{\equal{#1}{}}{g}{g_{#1}^{#2}}}
\newcommandx{\Phiverletq}[2][1=,2=]{\ifthenelse{\equal{#1}{}}{\widetilde{\Phi}}{\widetilde{\Phi}_{#1}^{\circ (#2)}}}
\newcommandx{\Phiverletqi}[2][1=,2=]{\ifthenelse{\equal{#1}{}}{\bar{\Psi}}{\bar{\Psi}_{#1}^{(#2)}}}
\newcommandx{\Pkerhmc}[2][1=,2=]{\ifthenelse{\equal{#1}{}}{\mathrm{P}}{\mathrm{P}_{#1, #2}}}
\newcommandx{\tPkerhmc}[2][1=,2=]{\ifthenelse{\equal{#1}{}}{\tilde{\mathrm{P}}}{\tilde{\mathrm{P}}_{#1, #2}}}
\newcommandx{\PkerhmcD}[2][1=,2=]{\ifthenelse{\equal{#1}{}}{\mathrm{K}}{\mathrm{K}_{#1, #2}}}

\def\argmax{\text{argmax}}
\newcommandx{\thetahat}[1][1=]{\ifthenelse{\equal{#1}{}}{\hat{\theta}}{\hat{\theta}(#1)}}
\newcommandx{\Kerpi}[2][1=,2=]{\ifthenelse{\equal{#1}{}}{\Pi}{\Pi_{#1}^{#2}}}
\newcommandx{\poinu}[2][1=,2=]{\ifthenelse{\equal{#1}{}}{\nu}{\nu_{#1}^{#2}}}
\def\lyapD{\mathpzc{V}}

\begin{document}

% If your paper is accepted and the title of your paper is very long,
% the style will print as headings an error message. Use the following
% command to supply a shorter title of your paper so that it can be
% used as headings.
%
\runningtitle{SA with Biased MCMC for EM}

% If your paper is accepted and the number of authors is large, the
% style will print as headings an error message. Use the following
% command to supply a shorter version of the authors names so that
% they can be used as headings (for example, use only the surnames)
%
%\runningauthor{Surname 1, Surname 2, Surname 3, ...., Surname n}

\twocolumn[

\aistatstitle{Stochastic Approximation with Biased MCMC for\\Expectation Maximization}

\aistatsauthor{ Samuel Gruffaz \And Kyurae Kim \And Alain Oliviero Durmus \And Jacob R. Gardner }

\aistatsaddress{ Centre Borelli, ENS Paris-Saclay \And  UPenn \And CMAP, CNRS, \'Ecole Polytechnique \And UPenn} ]

\begin{abstract}
  The expectation maximization (EM) algorithm is a widespread method for empirical Bayesian inference, but its expectation step (E-step) is often intractable.
  Employing a stochastic approximation scheme with Markov chain Monte Carlo (MCMC) can circumvent this issue, resulting in an algorithm known as MCMC-SAEM. 
  While theoretical guarantees for MCMC-SAEM have previously been established, these results are restricted to the case where asymptotically unbiased MCMC algorithms are used.
  In practice, MCMC-SAEM is often run with asymptotically biased MCMC, for which the consequences are theoretically less understood.
  In this work, we fill this gap by analyzing the asymptotics and non-asymptotics of SAEM with biased MCMC steps, particularly the effect of bias.
  We also provide numerical experiments comparing the Metropolis-adjusted Langevin algorithm (MALA), which is asymptotically unbiased, and the unadjusted Langevin algorithm (ULA), which is asymptotically biased, on synthetic and real datasets. 
  Experimental results show that ULA is more stable with respect to the choice of Langevin stepsize and can sometimes result in faster convergence.
\end{abstract}

\section{INTRODUCTION}

Probabilistic modeling with latent variables is an essential tool for modeling observational data generated from complex latent structures.
While eliciting priors for this class of models is often straightforward, \textit{e.g}, data generated from unobserved groups can be modeled using mixtures \citep{mclachlan2019finite}, while group-level variabilities can be modeled with mixed effects \citep{kuhn2005maximum}, how we should set the hyper-parameters $\theta$ is not always self-evident. 
For this, the empirical Bayes paradigm applies the \textit{maximum likelihood principle}~\citep{Robbins1956empirical,Efron2019bayes}.
That is, it infers the (hyper-)parameters $\theta$ from data by maximizing the marginal log-likelihood $l(\theta)\triangleq\log(p(y|\theta))$,
{%
\setlength{\belowdisplayskip}{.5ex} \setlength{\belowdisplayshortskip}{.5ex}
\setlength{\abovedisplayskip}{.5ex} \setlength{\abovedisplayshortskip}{.5ex} 
\begin{equation}
  \label{eq:original_problem}
  \hspace{-.7em}
  \argmax_{\theta \in \Theta}  l(\theta)=\argmax_{\theta \in \Theta} \log\left(\int_{\mcz} p(y,z|\theta)\mathrm{d}z\right)
\end{equation}
}%
where $y\in\mathcal{Y}\subset \Rset^{d_y}$ % ($d_y=N\times a$ if we have $N$ observations in $\Rset^a$
 are the observations, $z\in \mathcal{Z}\subset \Rset^{d_z}$ are the latent variables, and $\theta \in \Theta \subseteq \Rset^{d_\theta}$ denote the parameters.

Unfortunately, the marginal log-likelihood is often intractable, making the empirical Bayes problem hard.
As a solution, \citet{dempster1977maximum} have proposed the \textit{expectation-maximization} (EM) algorithm, which has subsequently been immensely successful in practice, and its convergence properties have been studied extensively~\citep{Mclachlan2007em}.
It is now known to converge to stationary points under mild conditions.
However, the canonical EM algorithm may not apply immediately in many practical cases.
For example, gradient ascent steps must be used when the maximization step (M-step) of EM is not available in closed form~\citep{baey2023efficient}.

This work focuses on the setting where the expectation step (E-step) is intractable.
In this case, the integral in~\Cref{eq:original_problem} needs to be numerically approximated using methods such as Monte Carlo (MC), importance sampling (IS), and Markov chain Monte Carlo (MCMC) (See the book by \citet{Robert2004monte} for an overview of these methods). 
Furthermore, when the model belongs to the exponential family such that the M-step can be computed in closed-form, \citet{delyon1999convergence} show that the EM algorithm can be reduced to approximating sufficient statistics generated from the sampling algorithm.
In particular, they leverage stochastic approximation (SA; \citealp{robbins1951stochastic}) for approximating sufficient statistics, resulting in the SAEM algorithm.
SAEM has been shown to be particularly efficient in practice, especially since most practical models can be embedded in an exponential family form~\citep{debavelaere2021curved}, and have been widely used through multiple popular software packages such as \texttt{NONMEM}~\citep{Bauer2019nonmem}, \texttt{saemix} \citep{Comets2017parameter}, \texttt{Monolix} \citep{lavielle2014mixed}, and many more.

Leveraging SA comes at a price: The performance of SAEM crucially depends on that of the underlying sampling algorithm.
However, most popular MCMC algorithms used for SAEM, such as the Metropolis-adjusted Langevin Algorithm (MALA; \citealp{roberts1996exponential}) and Hamiltonian Monte Carlo (HMC; \citealp{Duane1987hybrid}), perform poorly unless carefully tuned, especially in high dimensions.
While targeting an ``optimal'' acceptance rate is an effective way to tune these algorithms~\citep{Gilks1998adaptive}, this is not straightforward in the SAEM context: the limiting distribution of the Markov Chain is not fixed.

Note that the problem of tuning stems from the fact that we are employing Metropolis-Hastings adjustments.
Therefore, not adjusting at all, resulting in approximate MCMC methods such as the unadjusted Langevin algorithm (ULA), would immediately resolve this issue.
While these methods are ``approximate'' in the sense that their limiting stationary distribution is biased, tuning is less critical to their performance~\citep{Valentinde2019efficient}.
Furthermore, ULA has theoretically been shown to converge faster than its unbiased counterpart MALA~\citep{durmus2017nonasymptotic}.
Therefore, this work establishes theoretical guarantees for SAEM with biased MCMC.
Also, we compare the performance of MCMC-SAEM with MALA versus ULA on practical examples.
The results suggest that, in general, ULA can use larger stepsizes than MALA, resulting in faster convergence.

\vspace{-1ex}
\paragraph{Contributions} 
Our contributions are two-fold: the asymptotic analysis (\Cref{section:asymptotic}) and non-asymptotic analysis (\Cref{section:Non_asymptotic}) of MCMC-SAEM with biased MCMC.
For the asymptotic analysis, we generalize the analysis of~\citet{tadic2017asymptotic} on stochastic gradient optimization to the case of SA.
Furthermore, we explicitly control the asymptotic convergence according to the bias of MCMC through the smoothness of the problem.
This improves over the almost sure convergence results provided by \citet[\S V]{Dieuleveut2023stochastic}.
For our non-asymptotic analysis, we extend the framework~\citet{karimi2019non-asymptotic} to include asymptotic bias, and the guarantee is in high-probability.
This characterizes the effect of Markov chain concentration, unlike the convergence results in expectation provided by~\citet{karimi2019non-asymptotic,Dieuleveut2023stochastic}.

\section{BACKGROUND}
\label{section:MCMCSAEM}
% First, we expose the main ideas and assumptions related to the EM algorithm.
%  Secondly, we present its stochastic variants and the so-called MCMC-SAEM procedure.
% Thirdly, we motivate the comparison of asymptotically biased and unbiased MCMC methods within the SAEM procedure by introducing ULA and MALA.
% Finally, we give the SA formalism related to SAEM with biased MCMC.
 % S can be convex, Z is just open to , theta can be constrained
 
\vspace{-1ex}
\subsection{Expectation-Maximization}
\vspace{-1ex}
We assume that the joint distribution of the observations and latent variable belongs to the curved exponential family:
\begin{assumption}
  \label{assumption:exponential_family}
For any $y\in \mcy,z\in \mcz$ and $\theta\in \Theta$,
{%
\setlength{\belowdisplayskip}{1ex} \setlength{\belowdisplayshortskip}{1ex}
\setlength{\abovedisplayskip}{1ex} \setlength{\abovedisplayshortskip}{1ex} 
$$ p(y,z|\theta)=h(y,z)\exp(S(y,z)^{\top}\phi(\theta)-\psi(\theta))\eqsp , $$
}%
where $\Theta,\mathcal{Z}$, are open, $S:\mcy\times \mcz \to  \Rset^d $ is continuous, and $\phi: \Theta \to \Rset^d $ and $\psi: \Theta\to \Rset$ are continuously differentiable.
\end{assumption}
\vspace{-1ex}
In this paper, $y\in \mathcal{Y}$ is fixed, and the parameter $\theta$ is estimated through maximum marginal log-likelihood $l :\theta\in \Theta \mapsto \log p(y|\theta)$.
In most cases, $l$ cannot be maximized directly and is not even tractable.
EM solves this by employing the \textit{majorize-minimize principle}: (See, \textit{e.g.}, \citealp[Chapter 8]{lange2016mm}.)
An upper bound of $-l$ is derived using Jensen's inequality and a density $q$ on $\mcz$, for any $\theta \in \Theta$,
{%
\setlength{\belowdisplayskip}{1ex} \setlength{\belowdisplayshortskip}{1ex}
\setlength{\abovedisplayskip}{.5ex} \setlength{\abovedisplayshortskip}{.5ex} 
 \begin{align}
  &-\log p(y|\theta)=\log \left( \int_{\mcz} p(y,z|\theta) \times\frac{q(z)}{q(z)}\operatorname{dz} \right) \\
  &\leq -\int_{\mcz} \log\left(\frac{p(y,z|\theta)}{q(z)}\right)q(z) \operatorname{dz}=Q(\theta,q)\\
  &=-\left(\mathbb{E}_{z\sim q}(\log(p(y,z|\theta ))) \dz
  + \operatorname{Ent}(q)+\text{cst}(y)\right) \eqsp .
 \end{align}
 }%
 Denoting $D_\mcz \triangleq \{f\in \mathrm{L}^1(\mcz):\,f\geq 0,\, \int_\mcz f \dz=1 \}$, the function $q\in D_\mcz\mapsto Q(\theta,q)$ is minimized by $q^*(z)\triangleq p(z|y,\theta)$ such that $Q(\theta,q^*)=-l(\theta)$. 
 Thus, by considering
 $\theta^* = \argmin_{\theta \in \Theta} Q(\theta,q^*)$, we have $l(\theta^*)\leq l(\theta)$. 
 This procedure offers a recipe to construct a maximizing sequence of $l$.
 Moreover, under \Cref{assumption:exponential_family}, 
  $$\theta^* =\argmax_{\theta \in \Theta} \mathbb{E}_{z\sim p(z|y,\theta)}\left(S(y,z)\right)^{\top}\phi(\theta)-\psi(\theta)\eqsp . $$
 We introduce the following assumption to maximize this last expression:
\begin{assumption}
    \label{assumption:Maximization}
    Denoting $L(s,\theta)\triangleq s \cdot \phi(\theta)-\psi(\theta)$, there exists a function $\hat{\theta}:\Rset^d \to \Theta $, such that for any $(s,\theta)\in \Rset^d\times \Theta$,
    $$  \ee L(s,\theta)\leq L(s,\hat{\theta}(s))\eqsp .$$
\end{assumption}
\vspace{-2ex}
If necessary, this assumption can also be met by ``exponentializing'' as shown by~\citet{debavelaere2021curved}.

The EM algorithm is defined as follows:
Let $(s_k)_{k\geq 0},(\theta_k)_{k\geq 0} $ be initialized from $\theta_0\in \Theta $ and follow the recursion for any $k\geq 0$:
 \begin{enumerate}
\vspace{-1ex}
    \item[\ding{182}] \textbf{Expectation:} Denoting by $\bar{s}: \theta \in \Theta \mapsto \int_{\mcz} S(y,z)p(z|\theta,y)\dz$,
        set $s_k=\bar{s}(\theta_k)$.
    \item[\ding{183}] \textbf{Maximization:} Set $\theta_{k+1}=\hat{\theta}(s_k)$, which implies $l(\theta_{k+1})\geq l(\theta_k)$.
\vspace{-1ex}
 \end{enumerate} 
This algorithm quickly converges towards a local maximum of $l$ under mild conditions. (See the review by~\citet{Mclachlan2007em}.)

\subsection{EM as a Root Finding Problem}
The EM recursion can be seen as a fixed-point iterative scheme since, for any $k\geq 0$, $\theta_{k+1}=\hat{\theta}\circ \bar{s}(\theta_k)$.
 If the sequence $(s_n)$ converges towards $s^*$ and $\bar{s}\circ \hat{\theta}$ is continuous, we then have
{%
\setlength{\belowdisplayskip}{1ex} \setlength{\belowdisplayshortskip}{1ex}
\setlength{\abovedisplayskip}{1ex} \setlength{\abovedisplayshortskip}{1ex} 
\begin{equation} 
  h(s^*)=0,\eqsp \text{where}\eqsp h\triangleq\bar{s}\circ \hat{\theta}(\cdot)-\Id \eqsp .
\end{equation}
}%
Moreover, if $s^*$ is a root of $h $, then $\hat{\theta}(s^*)$ is a root of $\hat{\theta}\circ \bar{s}(\cdot)-\Id$.
And reciprocally, if $\theta^*$ is a root of $\hat{\theta}\circ \bar{s}(\cdot)-\Id $, then $\bar{s}(\theta^*)$ is a root of $h$.
This suggests that EM can be reduced to a problem of finding the root \(h\) with respect to \(s\).

\paragraph{EM as Stochastic Approximation}
Finding the root is more general than gradient descent as \(h\) may not be the \textit{gradient} of any known function.
Instead, it is the \textit{descent direction} for some Lyapunov function $V: s\in \Rset^d \mapsto -l\circ \hat{\theta}(s)$.
As such, the proposed scheme is part of the more general stochastic approximation (SA) framework \citep{Dieuleveut2023stochastic}.
The most basic form of the SA algorithm is described as follows:
\begin{align}
  s_{k+1}=s_k+\gamma_{k+1}H(s_k,Z_{k+1}),\;\; Z_{k+1}\sim p( \cdot |\theta_k,y)\eqsp ,
  \nonumber
\end{align}
where $H(s_k,Z_{k+1})=S(y,Z_{k+1})-s_k$ is a random oracle of $h(s_k)$ and $(\gamma_k)_k$ is a deterministic stepsize sequence.
Solving the EM problem using these iterates is known as SAEM algorithm~\citep{delyon1999convergence}. 

While we have motivated the reduction of optimizing \(l\) to finding the root of \(h\), it is not apparent the solutions such that $h=0$ and $\nabla l=0$ are equivalent.
The following lemma will clarify the relationship between \(h\), \(l\), and the Lyapunov function \(V\) given assumptions on the Hessian of $L$ at its maximums:
% In fact, $h$ is not the gradient of any function such that gradient-based methods convergence analysis does not apply, but the SA framework does.
%     In order to apply the main methodology of SA, we introduce the last regularity assumptions about the underlying probability model, which are also taken in \citep{Dieuleveut2023stochastic}.
   \begin{assumption}
    \label{hyp:regularity_loss}
    The functions $\hat{\theta},l(\cdot),\phi(\cdot),\bar{s}(\cdot)$ and $\psi(\cdot)$ are $p$-continuously differentiable with $p>d$.
    Moreover, denote
    \[
      A(s)\triangleq \partial_s \hat{\theta}(s)^{\top} \partial_\theta^2 L(s,\hat{\theta}(s)) \partial_s \hat{\theta}(s)\eqsp .
    \]
    Then, there exist $\lambda_m,\lambda_M>0$ such that for any $s\in \Rset^d$:
    \begin{equation}
        \label{eq:A_bound}
        \lambda_m|v|^2\leq \left\langle A(s)v|v \right\rangle\leq \lambda_M |v|^2,\ee v\in \Rset^d \eqsp .
    \end{equation}
\end{assumption}
    \begin{lemma}
      \label{lemma:preliminary}
      Under \Cref{assumption:exponential_family}, \Cref{assumption:Maximization} and \Cref{hyp:regularity_loss},  $V$ is $p$-continuously differentiable and verifies for any $s\in \Rset^d$,
      \begin{align}
        \label{eq:reg_V_h}
        &F(s)\triangleq \left\langle\nabla V(s)| h(s) \right\rangle\leq - \lambda_m |h(s)|^2\\
        & |\nabla V(s)|\leq \lambda_M |h(s)| \\
        & \mathsf{S}=\{s\in \Rset^d: F(s)=0\}=\{s\in \Rset^d: \nabla V(s)=0 \}\\
        & \hat{\theta}(\mathsf{S})=\{\theta \in \Theta: l(\theta)=0\}, \eqsp \operatorname{int}(V(\mathsf{S}))=\emptyset
      \end{align}
    \end{lemma}
    \begin{proof}
In the proof by \citet[Lemma 2]{delyon1999convergence}, they derive that for any $s\in \Rset^d$, $\nabla V(s)=-A(s)h(s)$, the results are then straightforward with the regularity assumption on $A$.
\vspace{-1.5ex}
\end{proof}
In other words, $h(s)$ is a proxy of $-\nabla V(s)$ for any $s\in \Rset^d$. 
This Lemma makes it clear that converging to $\mathsf{S}$ by working in the sufficient statistics space recovers the solutions of the original problem \eqref{eq:original_problem}.
From now on, we drop the dependence of $y$ in \(S(y,z)\) such that $S(z)\triangleq S(y,z)$ since we are interested only in $z$.

 %If $p(z|y,\theta)$ is not tractable or hard to sample because of very high dimensions, we can rely on a variational approxination by taking $q$ as a mean field and minimizae the right hand, which is nearly the ELBOW.

%Since the Expectation step $\bar{s}$ is not always computable, \cite{delyon1999convergence} propose to approximate $\bar{s}(\theta)$ using samples from $p( \cdot |\theta,y) $ for any $\theta\in \Theta$ obtained using Monte Carlo (MC) methods.
 
\paragraph{MCMC-SAEM}
SAEM can be generalized to models with intractable likelihoods by leveraging Markov chain Monte Carlo (MCMC) as proposed by~\citet{kuhn2004coupling}.
An MCMC algorithm form a Markov kernel $\Pi_\theta $, such that, for any $z\in\mcz$, $\lim_{n\to\infty}||\Pi_\theta(z,\cdot)^{(n)}-\pi_\theta ||_{\text{TV}}=0$, where
$\pi_\theta$ is a the target distribution related to $p(\cdot|\theta,y)$. 
It means that asymptotically, sampling from $ \Pi_\theta(z,\cdot)$ is nearly equivalent to sampling from $\pi_\theta$.
Then, the Markov chain is said to be asymptotically unbiased, contrary to the case $\lim_{n\to\infty}||\Pi_\theta(z,\cdot)^{(n)}-\pi_\theta ||_{\text{TV}}>0$, where we say that the chain is asymptotically biased.

Note that, in general, the chain is biased in finite time $||\Pi_\theta(z,\cdot)^{(n)}-\pi_\theta ||_{\text{TV}}>0 $ for any $n\geq 0$.
Despite this, MCMC-SAEM, MCMC applied to SAEM, generates a sequence $(\theta_k)_k$ that converges to a local maximum almost surely \citep{Dieuleveut2023stochastic,kuhn2004coupling}, where the $E$-step is replaced by 
$$s_{k+1}=s_k+\gamma_{k+1}(S(y,z_{k+1})-s_k),\ee z_{k+1} \sim \Pi_{\theta_k}(z_{k},\cdot)\eqsp, $$
for any $k\geq 0$, starting with $(s_0,z_0)\in \Rset^d\times \mcz$.
%  There is a mean-field drift $h: s\in \Rset^d \mapsto \bar{s}\left(\hat{\theta}(s)\right)-s$ behind this SA scheme. Indeed, formally $\mathbb{E}(s_{k+1}|s_k,z_k)=s_k+\gamma_k (\int_\mcz S(y,z)\Pi_{\theta_k}(z_k,\dz)-s_k) $ and $\lim_{k\to\infty}\int_\mcz S(y,z)\Pi_{\theta_k}(z_k,\dz)= \bar{s}(\theta_k)=\bar{s}(\hat{\theta}(s_k)) $ under mild conditions.
%   Thus, $(s_k)_k$ converges towards a stationnary point $s^*$ such ahat $h(s)=0$. Hopefuly,  the Lyapunov function $ V: s\in\Rset^d \mapsto l\circ \hat{\theta}(s) \Rset  $ is used

\subsection{MCMC-SAEM with Approximate MCMC Algorithms}
At the core of the practical performance of MCMC-SAEM is the choice of the MCMC algorithm. 
While MALA is often used, its asymptotically biased counterpart ULA has been shown to mix faster at high dimensions \citep{durmus2017nonasymptotic}.
This means that, in finite time, the total bias of ULA can be lower compared to the asymptotically unbiased MALA.
Therefore, while we consider a general biased Markov chain to study the convergence of MCMC-SAEM in \Cref{section:asymptotic} and \ref{section:Non_asymptotic}, we are interested in studying the difference between ULA and MALA.
Our experiments in \Cref{section:experiments} will exclusively focus on this question.
%That is why we detail these methods a little before going beyond.

\paragraph{ULA and MALA}
\label{para:ULAMALA}
The ULA Markov chain $(X_k)_{k\geq 0}$ is derived from the Euler–Maruyama discretization scheme associated with the Langevin diffusion related to the force $U\triangleq -\nabla \log(\pi)$ if $\pi$ is the target distribution, at iteration $k\geq 0$
{%
\setlength{\belowdisplayskip}{1.5ex} \setlength{\belowdisplayshortskip}{1.5ex}
\setlength{\abovedisplayskip}{1.5ex} \setlength{\abovedisplayshortskip}{1.5ex} 
\begin{equation}
     X_{k+1}=X_k-\eta_{k+1} \nabla U(X_k)+\sqrt{2\eta_{k+1}}Z_{k+1}
\end{equation}
}%
where $(Z_k)_{k\geq 0}$ is an i.i.d sequence of standard Gaussian $d$-dimensional random vectors and $(\eta_k)_{k\geq 1}$ is a sequence of stepsize, which can be either constant or decrease to 0.
The MALA Markov chain follows the same recursion by adding an acceptation/rejection step of $X_{k+1}$, making the chain asymptotically unbiased. 

The properties of MALA chains have been studied in depth in \citep{roberts1996exponential}, where it is shown that the chain converges geometrically fast (geometric ergodicity) under mild assumptions on the tail of the target distribution.
This is key since geometric ergodicity is the main assumption of the MCMC-SAEM convergence theorem \citep{kuhn2004coupling}.
Despite being geometrically ergodic, the adjustment step of MALA can become a curse in high dimensions.
To maintain a sufficient level of acceptance, a smaller stepsize must be used, resulting in a slow mixing chain. (Mixing is a fundamentally non-asymptotic notion, unlike geometric ergodicity.)

On the other hand, for ULA, denoting its limiting distribution as $\pi^\eta$, the mixing rate improves with $\eta$ \citep{durmus2017nonasymptotic} at the cost of increasing $||\pi^\eta-\pi ||_{\text{TV}}$ \citep{Valentinde2019efficient}.
This means one can trade off asymptotic bias for a faster mixing rate.
Considering this, we will expand the existing analysis, where $\eta$ is no longer a stepsize but a ``knob'' that controls bias, which can vary across the iterations.

\subsection{Stochastic Approximation with Biased Dynamics}
To incorporate biased MCMC chains into our analysis, we slightly modify the SA formalism \citep{Dieuleveut2023stochastic} to allow some freedom on the bias parameter $\eta$.

Let $(\gamma_n)_n$ and $(\eta_n)_n$ be two monotone nonincreasing sequences with $\gamma_0,\eta_0\in [0,1]^2$.
Define the nonhomogeneous Markov chain $\left\{Y_n^\gamma=(Z_n,S_n)\right\}_n$ on $\mcz\times \Rset^d$ as follows:
 Let $s_0=\theta \in \Rset^d, \ee z_0=z\in \Rset^d $ and for $n\geq 0$,
{%
\setlength{\belowdisplayskip}{1.5ex} \setlength{\belowdisplayshortskip}{1.5ex}
\setlength{\abovedisplayskip}{1ex} \setlength{\abovedisplayshortskip}{1ex} 
\begin{align}
    Z_{n+1} &\sim \Pi_{s_n}^{\eta_{n+1}}(Z_n, \cdot ) \nonumber \\
    s_{n+1} &=s_n+\gamma_{n+1} H(s_n,Z_{n+1}) 
    \eqsp,
    \label{eq:def_inhomogeneous_markov_chain}
\end{align}
}%
where $\{\Kerpi[s][\eta]=\Kerpi[\thetahat(s)][\eta], s,\eta\in \Rset^d\times (0,\eta_0]\} $ is a family of Markov transition probabilities and $H: s,z\in \Rset^d\times \mcz\mapsto S(z)-s $ is a field which satisfy the following conditions:
\begin{assumption}
    \label{hyp:reg_kernel}
    For any $s \in \Rset^d $ and $\eta \in (0,\eta_0]$, the Markov kernel $\Kerpi[s][\eta]$ has a single stationary distribution $\pi_{\thetahat[s],\eta}$, also denoted as $\pi_{s,\eta}$, such that $ \pi_{s,\eta}\Kerpi[s][\eta]=\pi_{s,\eta} $.
    In addition, $H:\Rset^d\times \mcz\to \Rset^d$ is measurable for all $s\in \Rset^d$
    and $\int_\mcz |H(s,z)| \pi_{s,\eta }(\dd z)<\infty$ .
    \vspace{-1ex}
\end{assumption}
By considering the filtration $\left\{\mathcal{F}_n=\sigma(s_0,Z_i,i\leq n)\right\}_{n}$, 
the following decomposition clarifies the deterministic and stochastic parts of the dynamic:
{%
\setlength{\belowdisplayskip}{-2ex} \setlength{\belowdisplayshortskip}{-2ex}
\setlength{\abovedisplayskip}{1ex} \setlength{\abovedisplayshortskip}{1ex} 
\begin{align}
    &H(s_n,Z_{n+1}) = h(s_n)+\xi_n \eqsp, \\
    &\xi_n = e_n+\beta_n \eqsp,  \\
    &e_n = H(s_n,Z_{n+1})-\mathbb{E}_{W\sim\pi_{s_n,\eta_{n+1}} }\left[ H(s_n,W)\right]\eqsp, \\
    &\beta_n = \mathbb{E}_{W\sim\pi_{s_n,\eta_{n+1}} }\left[ H(s_n,W)\right]-h(s_n)\eqsp , \\
\end{align}
}%
where $h(s_n)$ is the mean field drift, $e_n$ is the Markovian noise and $(\beta_n)_n$ the bias.
 In the MCMC-SAEM context, this becomes
{%
\setlength{\belowdisplayskip}{1.5ex} \setlength{\belowdisplayshortskip}{1.5ex}
\setlength{\abovedisplayskip}{1.5ex} \setlength{\abovedisplayshortskip}{1.5ex} 
\begin{align}
&e_n=S(z_{n+1})-\check{s}_{\eta_{n+1}}(\thetahat[s_n]) \eqsp, \\
&\eqsp \beta_n=\check{s}_{\eta_{n+1}}(\thetahat[s_n])-\bar{s}(\thetahat[s_n])\eqsp ,
\end{align}
}%
where, for any $s\in \Rset^d$ and $\eta \in (0,\eta_0]$,
{%
\setlength{\belowdisplayskip}{0ex} \setlength{\belowdisplayshortskip}{0ex}
\setlength{\abovedisplayskip}{1ex} \setlength{\abovedisplayshortskip}{1ex} 
\begin{align}
h(s) &=\bar{s}(\thetahat[s])-s \eqsp, \\
\check{s}_\eta(\thetahat[s]) &=\int S(z) \pi_{s,\eta}(\dd z) \eqsp.
\end{align}
}%
This way, for any $s\in \Rset^d$, $\check{s}_\eta(\thetahat[s])$ is the biased approximation of $\bar{s}(\thetahat[s]) $. 
Lastly, we assume that the asymptotic bias is finite:
\vspace{1ex}
\begin{assumption}
    \label{hyp:bias_control}
{%
\setlength{\belowdisplayskip}{-1ex} \setlength{\belowdisplayshortskip}{0ex}
\setlength{\abovedisplayskip}{0ex} \setlength{\abovedisplayshortskip}{0ex} 
    $\limsup_{n\to \infty} |\beta_n|=\beta<+\infty$
}%
\end{assumption}
This assumption is reasonable if the growth of $z\mapsto S(z)$ can be mitigated by the tail decay of $\pi_{\theta,\eta}$ and $\pi_\theta$.

\section{ASYMPTOTIC ANALYSIS}
\label{section:asymptotic}
We estimate how the asymptotic bias of MCMC impacts the asymptotic convergence of $(s_n)$ generated by the recursion \eqref{eq:def_inhomogeneous_markov_chain}.
The analysis takes most of its arguments from \citep{tadic2017asymptotic} where the framework is slightly different, $h=\nabla f$ where $f$ has the same regularity that in \Cref{hyp:regularity_loss}.
We extend their results by using \Cref{lemma:preliminary}, i.e., the mean field is not a gradient but the proxy of a Lyapunov gradient.
The results are not specific to the SAEM scheme but can be generalized to other SA schemes as \Cref{lemma:preliminary} applies, but this is beyond the scope of this paper.

\subsection{Technical Assumptions} 
All the assumptions presented in this part are also used in \citep{tadic2017asymptotic}.
We take the following assumptions on $(\gamma_n)$, $(e_n)$:
 \begin{assumptionsup}
     \label{hyp:gamma_n}
     $\limsup _{n \rightarrow \infty}\left|\gamma_{n+1}^{-1}-\gamma_n^{-1}\right|<\infty$, $\lim_{n\to\infty}\gamma_n =0$ and $\sum_{n=0}^\infty \gamma_n=\infty$.
 \end{assumptionsup}
 Denoting by $a(n,t)=\max \left\{k\geq n:\sum_{i=n}^{k-1}\gamma_i\leq t \right\} $ for $n\geq 0$ and $t\in(0,\infty)$, $a(n,t)$ is well defined thank to \Cref{hyp:gamma_n}. 
 Furthermore, we assume:
 \begin{assumptionsup}
     \label{hyp:noise_bias_control}
     $(e_n)$ and $(\beta_n) $ are $\Rset^d$-valued stochastic processes satisfying,
     {%
     \setlength{\belowdisplayskip}{1ex} \setlength{\belowdisplayshortskip}{1ex}
     \setlength{\abovedisplayskip}{0ex} \setlength{\abovedisplayshortskip}{0ex} 
     \begin{equation}
         \lim_{n\to \infty} \max_{n\leq k <a(n,t)} \abs{\sum_{i=n}^k \gamma_i e_i} =0
     \end{equation}
     }%
     almost surely on $\{ \sup_n |s_n|<\infty\}$ for any $t\in(0,\infty)$.
 \end{assumptionsup}
 \Cref{hyp:noise_bias_control} can be established from assumptions of geometric ergodicity on the Markov kernels $\{\Pi_s^{\eta},\eqsp s\in\Rset^d,\eqsp \eta \in (0,\eta_0]\} $ and by controlling the growth of the sufficient statistics.
 Due to its technicalities, this development is relegated to \Cref{proof:simplification}.

 In numerous cases, we can even make the following assumption,
 \begin{assumptionsup}
     \label{hyp:analytic}
     $V(\cdot)$ is real analytic on $\Rset^d$.
 \end{assumptionsup}
This implies that $V$ can be locally represented by a power series.

\subsection{Main Result} 
 For any compact $Q\subset \Rset^d$, $\Lambda_Q$ denotes the event $\cup_{n=0}^\infty \cap_{k=n}^\infty \{\theta_k\in Q \}$, such that on $ \Lambda_Q$, $\{\theta_k\}$ is bounded and the bound is controlled by $Q$.
 The main result on the asymptotic bias of the recursion \eqref{eq:def_inhomogeneous_markov_chain} can be stated as follows:

 \begin{theorem}
     \label{thm:asymptotic}
      Suppose that \Cref{assumption:exponential_family}-\ref{hyp:bias_control} , \Cref{hyp:gamma_n}-\ref{hyp:noise_bias_control}. Let $Q \subset \mathbb{R}^{d}$ be any compact set. Then, the following are true:
      \vspace{-2ex}
      \begin{enumerate}[label=(\Roman*),wide, labelindent=0pt]
 \item \label{thm:asymp:item1} There exists a (deterministic) non-decreasing function $\psi_Q$ : $[0, \infty) \rightarrow$ $[0, \infty)$ 
 (independent of $\eta$ and depending only on $V(\cdot)$ ) such that $\lim _{t \rightarrow 0} \psi_Q(t)=\psi_Q(0)=0$ and
{%
\setlength{\belowdisplayskip}{1ex} \setlength{\belowdisplayshortskip}{1ex}
\setlength{\abovedisplayskip}{1ex} \setlength{\abovedisplayshortskip}{1ex} 
 $$
 \limsup _{n \rightarrow \infty} d\left(s_n, \mss \right) \leq \psi_Q(\beta)
 $$
}%
 almost surely on $\Lambda_Q$.
 \item \label{thm:asymp:item2} There exists a real number $K_Q \in(0, \infty)$ (independent of $\beta$ and depending only on $V(\cdot))$ such that
{%
\setlength{\belowdisplayskip}{1ex} \setlength{\belowdisplayshortskip}{1ex}
\setlength{\abovedisplayskip}{1.5ex} \setlength{\abovedisplayshortskip}{1.5ex} 
 \begin{align}
 &\limsup _{n \rightarrow \infty}\left\|\nabla V\left(s_n\right)\right\| \leq K_Q \beta^{q / 2}\eqsp ,\\
 & \quad \limsup _{n \rightarrow \infty} V\left(s_n\right)-\liminf _{n \rightarrow \infty} V\left(s_n\right) \leq K_Q \beta^q
 \end{align}
 }%
 almost surely on $\Lambda_Q$, where $q=\left(p-d\right) /(p-1)$.
 \item \label{thm:asymp:item3}If $V(\cdot)$ satisfies \Cref{hyp:analytic} there exist real numbers $r_Q \in(0,1), L_Q \in(0, \infty)$ (independent of $\beta$ and depending only on $V(\cdot))$ such that
{%
\setlength{\belowdisplayskip}{1ex} \setlength{\belowdisplayshortskip}{1ex}
\setlength{\abovedisplayskip}{1ex} \setlength{\abovedisplayshortskip}{1ex} 
 \begin{align}
  &\limsup _{n \rightarrow \infty}\left\|\nabla V\left(s_n\right)\right\| \leq L_Q \beta^{1 / 2} \\
  &\limsup _{n \rightarrow \infty} d(V\left(s_n\right), V(\mss )) \leq L_Q \beta \\
  &\quad \limsup _{n \rightarrow \infty} d\left(s_n, \mss\right) \leq L_Q \beta^{r_Q}\eqsp .
 \end{align}
 }%
 almost surely on $\Lambda_Q$.
 \end{enumerate}
 \end{theorem}
 \begin{proof}
    See \Cref{section:proof_biaspropagation} for the proof.
 \end{proof}
\Cref{thm:asymptotic}-\ref{thm:asymp:item1} formalizes the intuition that $\lim_{n\to \infty} d\left(s_n, \mss \right)=0$ if $\beta=\limsup_{n\to \infty} |\beta_n|\to0 $.
If $(\eta_n)$ encodes the stepsize of ULA, the last condition can be deduced from $\lim_{n\to \infty} |\eta_n|=0 $, since a smaller stepsize decreases bias.
Even though it is the case that $\beta>0$ in practice since $\eta_n=\eta>0$ is fixed, we can quantify the impact of the bias on $V$ in \Cref{thm:asymptotic}-\ref{thm:asymp:item2}. 
Note the impact of the regularity of $V$ according to the dimension encoded in $q=\left(p-d\right) /(p-1)$: at the limit when $p\to \infty$, we recover $q=1$ as in \Cref{thm:asymptotic}-\ref{thm:asymp:item3}.
Therefore, the impact of the bias $\beta$ on $\limsup _{n \rightarrow \infty}\left\|\nabla V\left(s_n\right)\right\|$ is smoothed by the regularity of $V$.
For the case of $\beta=0$, we recover the convergence of the sequence towards stationary points $\mss$, which has already been established by \citet{kuhn2004coupling}. 

Remark that \Cref{thm:asymptotic} is local in the sense that it holds on the event $\Lambda_Q$ and not globally. (\textit{i.e.}, there exists a compact set $Q$ such that $\Lambda_Q$ holds almost surely.)
Previous results have assumed that the sequence $(s_n)$ is bounded to make their result global and have argued that we can bound it by design through ``reinitialization'' as studied by \citet{kuhn2004coupling, andrieu2005stability}.
Here, we did not use the recursion proposed \cite[p.9]{andrieu2005stability} to be more aligned with practical implementations.
 
The constants $K_Q,L_Q$ depend explicitly on the bounds on $|\nabla V| $, $|\nabla^2 V| $ (the maximal singular value of $A$ to be precise) and the Yomdin and Lojasiewicz constants applied to
$V$.
(See, \textit{e.g}, Proposition 8.1 and 8.2 by \citet{tadic2017asymptotic}, which are generalizations of Sard's Theorem.)
The explicit forms of $K_Q,L_Q$ are given at the end of the proof in \Cref{section:proof_biaspropagation}.

\section{NON ASYMPTOTIC ANALYSIS}
\label{section:Non_asymptotic}
For the non-asymptotic analysis, we assume the bias parameters are fixed for any $n\geq 0, \eqsp \eta_n\in (0,\eta_0]$.
Furthermore, we rely on assumptions typical for non-asymptotic analysis of SA.
Our main result is a non-asymptotic high probability convergence guarantee for SAEM with biased MCMC.

%the limiting factor is due to the kernel regularity because it is statedependant
% If it is not state dependant, we can recover a term of order 1/() in the bound of the poisson solution

\subsection{Technical Assumptions} 
First, we impose assumptions on the MCMC kernel, which are standard in the non-asymptotic analysis of stochastic approximation with state-dependent Markovian noise. 
\begin{assumptionnona}\label{assumption:bounded_update}
The update is bounded by a constant \( 0 < \sigma < \infty\) as
{%
\setlength{\belowdisplayskip}{-2ex} \setlength{\belowdisplayshortskip}{-2ex}
\setlength{\abovedisplayskip}{.ex} \setlength{\abovedisplayshortskip}{.ex} 
\[
   \sup_{(s,z) \in \Rset^d\times \mathcal{Z}} | H\left(s, z\right) - h\left(s\right) | \leq \sigma\eqsp .
\]
}%
\end{assumptionnona}
While this assumption is quite strong, it is necessary to control the properties of the solution to the Poisson equation as in the following assumption:
 
\begin{assumptionnona}\label{assumption:poisson_equation}
For any $s\in \Rset^d$, there exists a solution \(\nu_s^{\eta}: \mcz \to \Rset^d\) to the Poisson equation such that
\[
   \nu_s^{\eta} - \Pi^{\eta}_{s} \nu_s^{\eta}
   =
   S\left(\cdot\right) - \check{s}_{\eta}(\thetahat[s])\eqsp ,
\]
where for any $z\in\mcz$, $\Pi^{\eta}_{s} \nu_s^{\eta}(z)=\int_\mcz \nu_s^{\eta}(z) \Pi^{\eta}_{s}(z,\dz) $ and for any $s\in\Rset^d$, $\check{s}_{\eta}(\thetahat[s])=\int_\mcz S(z)\dd \pi_{s,\eta} $. Moreover, there exist some bound $L_{\nu}^{(0)}, L_{\nu}^{(1)}>0 $ such that
\begin{align}
   &\sup_{(s,z) \in \Rset^d\times \mathcal{Z}}\left\{  | \nu_s^{\eta}(z) |,  | \Pi^{\eta}_{s} \nu_s^{\eta}(z) | \right\} \leq L_{\nu}^{(0)}\eqsp ,\\
   &\sup_{(s,z) \in \Rset^d\times \mathcal{Z}} | \Pi_s^{\eta} \nu_s^{\eta}(z) - \Pi_{s^{\prime}}^{\eta} \nu_{s^{\prime}}^{\eta}(z) | \leq L_{\nu}^{(1)} | s - s' |\eqsp .
\end{align}
\end{assumptionnona}
Remark that \Cref{assumption:bounded_update}-\ref{assumption:poisson_equation} implies \Cref{hyp:reg_kernel}.
Also, under \Cref{assumption:poisson_equation},
\begin{equation}
  e_n=S(Z_{n+1})-\check{s}_{\eta}(\thetahat[s_n])=\nu_s^{\eta}(Z_{n+1}) - \Pi^{\eta}_{s} \nu_s^{\eta}(Z_{n+1})\eqsp ,
\end{equation}
which is crucial for the analysis in the Supplementary material \Cref{section:proof_nonasymptotic}. 
This type of assumption has first been used by \cite{karimi2019non-asymptotic} and has since been standard in the analysis of stochastic approximation with state-dependent Markovian noise.
See, \textit{e.g}, the works of~\citet[Assumption 3.7]{Alacaoglu2023convergence},~\citet[Assumption 2.4]{Roy2022constrained} for some recent examples.
If we assume that \textbf{(i)} both \(\Pi_s^{\eta}\), \(H\left(s, z\right)\) are uniformly Lipschitz with respect to \(s\) for any \(z \in \mathcal{Z}\), \textbf{(ii)} \(\Pi\) is uniformly geometrically ergodically converging, and \textbf{(iii)} \cref{assumption:bounded_update} holds, \citet[Lemma 7]{karimi2019non-asymptotic} establish \(L_{\nu}^{(0)}, L_{\nu}^{(1)}\) explicitly.

% Given the solution to the Poisson equation, we are interested in the effect of the \textit{concentration} of the Markov chain states.
% To analyze this, we assume that the Markov chain kernel has sub-Gaussian tails:
% \begin{assumptionnona}\label{assumption:subgaussian}
% There exists $c>0$ such that the MCMC kernel has sub-Gaussian tails for \(\nu_{s}^{\eta}\), i.e., for any $s\in\Rset^d$,
%  \[
%    \mathbb{E}^{\mathcal{F}_{k+1}}\exp\left( {\left| \nu_{s}^{\eta}(Z_{k+1}) - \Pi_{s}^{\eta} \nu_{s}^{\eta}(Z_k) \right|}^2 / c^2 \right)
%    \leq
%    \exp\left(1\right).
%  \]
%  \end{assumptionnona}
%  This assumption is satisfied by most MCMC kernels that involve some sort of Gaussian proposal mechanism.
 In this work, we are particularly interested in asymptotically biased MCMC algorithms:
 \begin{assumptionnona}\label{assumption:asymptotic_bias}
 The asymptotic bias of the MCMC kernel is bounded for some \(0 \leq \tau_0, \tau_1 < \infty\) as
 \[
   \abs{\beta_k}^2 \leq \tau_0 + \tau_1 \abs{h\left(s_k\right)}^2\eqsp .
 \]
 Moreover, under \Cref{hyp:regularity_loss}, $C_{b_1}\triangleq \lambda_M \left( \frac{1}{2} \sqrt{\tau_0}  + \sqrt{\tau_1} \right)<\lambda_m $, where $\lambda_m,\lambda_M>0 $ are defined in \Cref{hyp:regularity_loss}.
 \end{assumptionnona}
 This assumption has been used by \cite{Dieuleveut2023stochastic}, and encompasses both iterate-dependent and -independent bias.
  It is a refinement of \Cref{hyp:bias_control}. 
 The condition about $ C_{b_1}$ is specific to the SAEM framework; the bias is bounded by the matrix conditioning of $A$ given in \Cref{hyp:regularity_loss}. This reveals crucial in the high probability bound.
 
 The remaining assumptions are standard in the non-asymptotic analysis of stochastic approximation.
 (See H1-2 in the recent review by \cite{Dieuleveut2023stochastic}.)
 \begin{assumptionnona}
  \label{assumption:smoothness}
 The Lyapunov function \(V\) is smooth and bounded below such that
 \begin{align}
   \abs{ \nabla V\left(s\right) - \nabla V\left(s^{\prime}\right) } &\leq L_V \,\abs{ s - s^{\prime} }
   \qquad
   V\left(s\right) \geq V^*
 \end{align}
 for \(V^* = \inf_{s \in \Rset^d} V\left(s\right) > - \infty\) and some \(0 < L_V < \infty\).
 \end{assumptionnona}
 
%  \begin{assumption}\label{assumption:lyapunov_meanfield}
%  There exists some \(0 \leq \lambda_M < \infty\) such that the norm of the Lyapunov function is bounded by that of the mean-field as
%  \[
%    \abs{ \nabla V\left(s\right) } \leq \lambda_M \abs{h\left(s\right)}.
%  \]
%  \end{assumption}
 
%  \begin{assumption}\label{assumption:coherence}
%  There exists some \(0 < c < \infty\) such that the gradient of the Lyapunov function is coherent with the mean-field as
%  \[
%    \left\langle \nabla V\left(s\right) \middle| h\left(s\right) \right\rangle
%    \leq -\lambda_m \abs{ h\left(s\right) }^2.
%  \]
%  \end{assumption}
 For the EM setting, \cref{assumption:smoothness} is problem-dependent and can not be recovered directly from \Cref{hyp:regularity_loss}.
 
\newpage
\subsection{Main Result}

 \begin{theorem}\label{thm:nonasymptotic}
 Assume \Cref{assumption:exponential_family}-\ref{hyp:regularity_loss},
 \cref{assumption:bounded_update}-\ref{assumption:asymptotic_bias} and
 \Cref{assumption:smoothness},
 
 Then, given a stepsize satisfying with $\alpha_1, \alpha_2>0$,
 \begin{align}
  \label{eq:stepsize_condition}
   &\gamma_{k+1} \leq \gamma_k,
   \quad
   \gamma_{k} \leq \alpha_1 \gamma_{k+1},
   \\
   &\gamma_{k} - \gamma_{k+1} \leq \alpha_2 \gamma_{k+1}^2,
   \quad
   \gamma_{0} \leq \frac{1}{2} \left( \lambda_m - C_{b_1} \right) / C_{n_1},
 \end{align}
 with probability at least \(1 - \delta\), the MCMC-SAEM algorithm guarantees that
 \begin{align}
   &\min_{k=0, \ldots, n} \abs{h\left(s_{k}\right)}^2
   \leq \frac{2}{\lambda_m - C_{b_1}} \times
   \nonumber
   \\
   &
   \left(
   \frac{
     V(s_0) - V^* + C_0 + \log \frac{1}{\delta} + C_{n_2} \sum^n_{k=0} \gamma_{k+1}^2  
   }{
     \sum_{k=0}^n \gamma_{k+1}
   }
   +
   C_{b_2}
   \right),
  \label{eq:results_high_proba}
 \end{align}
 where the constants are 
 \begin{align}
   C_0      
   &= L_{\nu}^{(0)} \left( \gamma_0 + 2 \lambda_M \right),
   \qquad
   \\
   C_{b_1}
   &= 
   \lambda_M \left( \frac{1}{2} \sqrt{\tau_0}  + \sqrt{\tau_1} \right), \\
   C_{b_2}
   &= 
   \frac{1}{2} \lambda_M \sqrt{\tau_0},
   \\
   C_{n_1}
   &= 
   L_{\nu}^{(1)} \lambda_M \sigma + L_V L_{\nu}^{(0)} \lambda_M \left(1+\sigma\right) + L_V \sigma^2,
   \\
   C_{n_2}
   &= 
   (2\lambda_M L_{\nu}^{(0)})^2 
   + L_{\nu}^{(1)} \lambda_M \left(\frac{1}{2} + \alpha_1 \sigma + \alpha_1 \frac{1}{2} \right) \\
   &\qquad\quad+ L_{\nu}^{(0)} \lambda_M \left( L_V + \alpha_2 + 1\right)
   + L_V.
 \end{align}
 \begin{proof}
     See \Cref{section:proof_nonasymptotic} for the proof.
 \end{proof}
 
 \end{theorem}
If we set $\gamma_{k}=\left( \lambda_m - C_{b_1} \right)/2C_{n_1}\sqrt{k}$ for any $k\geq 1$, then the stepsize satisfies \eqref{eq:stepsize_condition} with $\alpha_1,\alpha_2=\sqrt{2},(\sqrt{2}-1)/\sqrt{2}\gamma_1$, and thus the right hand in \eqref{eq:results_high_proba} becomes $\mathcal{O}(\log(n)/n+2C_{b_2}/\left(\lambda_m - C_{b_1}\right)) $.
The results are coherent with \Cref{thm:asymptotic}-\ref{thm:asymp:item3}, the bias is proportional to $\lambda_M \sqrt{\beta} $ as $n\to \infty$.

The key step in establishing the high-probability bound is to ensure the non-asymptotic transient bias of the Markov chain concentrates.
This is done by constructing a Martingale following the strategy of \citet{karimi2019non-asymptotic}, while the concentration inequality is by~\citet[Lemma 1]{Li2020Ahigh}.
Although this Lemma relied on Gaussian tails, this automatically follows from \cref{assumption:poisson_equation}.
%Under 
%which is conventional in SA with state-dependent Markovian noise.
Therefore, with this set of assumptions, SA with MCMC is well-behaved under a concentration perspective.

%Specifying non-asymptotic constants according to specific Markov chain parameters would be interesting if reasonably low constants exists in geometric ergodicity results, which is still awaited in MCMC theory.

Overall, as in \Cref{section:asymptotic}, we conclude that the regularity of the problem determines the limit of biased MCMC within SAEM.
 
%[Dire quelque chose sur la verification des hypothèses qui est technique]
% [dans les expériences faire une référence à la régularité ?]

\newpage
\vspace{-1ex}
\section{EXPERIMENTS}
\vspace{-1ex}
\label{section:experiments}

We will now empirically compare the performance of MCMC-SAEM with approximate and asymptotically exact MCMC kernels.
In particular, we compare ULA and MALA.
While the computation cost is comparable--MALA is slightly more expensive as it requires an additional evaluation of the unnormalized target density--their practical behavior can be different, as we will see in the experiments.
All experiments were implemented in the Julia language~\citep{Bezanson2017julia}.
For the stepsize, we use \(\gamma_k = 1 / \sqrt{k}\) for all experiments.
Furthermore, in the E-step of each SAEM iteration, 4 MCMC steps are performed as burn-in unless stated otherwise.
The source code used for the experiments is publically available online\footnote{Link to \textsc{GitHub} repository: \url{https://github.com/Red-Portal/MCMCSAEM.jl}}.
 
\vspace{-1ex}
\subsection{Logistic Regression on a Synthetic Dataset}
To illustrate the difference in the behavior of ULA and MALA, we first consider a toy problem.

%  \begin{figure}
%    \centering
%    \includegraphics[]{figures/logistic_error_01}
%    \caption{\textbf{Number of iterations until hitting a parameter \(\theta_k\) satisfying \(\abs{\theta_k - \theta^*}^2 \leq 10^{-3}\) versus the MALA/ULA stepsize.} 
%    The confidence intervals are 80\% quantiles computed from 32 independent replications.
%    }\label{fig:logistic}
% \vspace{-2ex}
%  \end{figure}
 
 \begin{figure}
     \centering
     \vspace{-1ex}
     \includegraphics{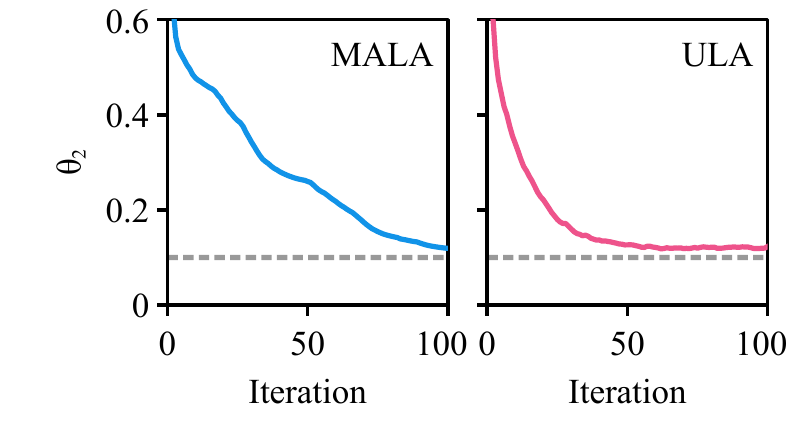}
     \vspace{-2ex}
     \caption{
     \textbf{Trajectory of the MCMC-SAEM iterates for \(\theta_2\) with a large MALA/ULA stepsize of \(\eta = 5 \times 10^{-3}\).}  
     MALA only makes ``occasional'' progress due to rejections, while ULA makes progress nonetheless, albeit with some asymptotic bias.
     The dotted line marks the true value \(\theta_2^*\).
     } 
     \label{fig:logistic_trajectory}
\vspace{-2ex}
 \end{figure}

%  \begin{figure}
%    \centering
%    % This requires \usepackage{subfloat}
%    \subfloat[Colon]{
%      \includegraphics[]{figures/logisticard_mlpd_01.pdf}
%    }
%    \subfloat[Prostate]{
%      \includegraphics[]{figures/logisticard_mlpd_02.pdf}
%    }
%    \subfloat[Leukemia]{
%      \includegraphics[]{figures/logisticard_mlpd_03.pdf}
%    }
%    \caption{\textbf{Mean log-predictive density (MLPD) resulting from the hyperparameters found by the algorithm versus the ULA stepsize.} 
%    The confidence intervals are 80\% bootstrap confidence intervals of the mean computed from 32 independent replications.}
%  \end{figure}
 
\vspace{-1ex}
\paragraph{Model}
 The model is a typical logistic regression model with a Gaussian prior on the coefficients:
{%
\setlength{\belowdisplayskip}{1ex} \setlength{\belowdisplayshortskip}{1ex}
\setlength{\abovedisplayskip}{1ex} \setlength{\abovedisplayshortskip}{1ex} 
 \begin{align}
    \beta &\sim \mathcal{N}\left(\mu \mathbf{1}_d, \sigma^2 \mathbf{I}_d\right)
    \\
    p_i &= \operatorname{logistic}\left({\beta^{\top}}{x_i}\right)\eqsp \\
    y_i   &\sim \mathsf{Bernoulli}\left(p_i\right)\eqsp . 
 \end{align}
}%
 We optimize for the hyperparameters \(\theta = (\mu, \sigma)\in \Rset\times \Rset_{>0}\).
 The number of datapoints is 1000, while the dimensionality of \(\beta\) is \(100\).
 The regression matrix is randomly generated to have a condition number of \(\kappa = 1000\).
 We initially run the respective MCMC algorithm for 10 iterations as burn-in, and then run MCMC-SAEM for \(100\) iterations.
 The true parameter is \(\theta^* = (1, 0.1)\), while we initialize the algorithm with \(\theta_0 = (0, 1)\).
 The MCMC chain is initialized from a standard Gaussian.
 
\vspace{-1ex}
\paragraph{Results}
For a small \(\eta\) (ULA/MALA stepsize), ULA and MALA perform similarly since MALA reduces to ULA.
On the other hand, they start to behave differently in the large \(\eta\) regime, as illustrated in \Cref{fig:logistic_trajectory}: MALA starts to reject more proposals, which results in slower convergence.
On the other hand, ULA does not reject anything, making constant progress.
This illustrates that, in the large-\(\eta\) regime, the asymptotic bias of ULA becomes less critical since MALA suffers from non-asymptotic transient bias from the rejections.

\subsection{Pharmacokinetics}
A popular application of MCMC-SAEM is nonlinear mixed modeling, which often arises in longitudinal studies with nonlinear models of progression.
Here, we will consider modeling the pharmacokinetics of Theophylline, which is a drug for respiratory diseases~\citep{Davidian1995nonlinear,Pinheiro1995approximations}.

\paragraph{Model}
Following the formulation of \citet{kuhn2005maximum}, we assume the concentration of the drug on the \(i\)th patient at the \(j\)th measurement can be modeled as
{%
\setlength{\belowdisplayskip}{1ex} \setlength{\belowdisplayshortskip}{1ex}
\setlength{\abovedisplayskip}{1ex} \setlength{\abovedisplayshortskip}{1ex} 
\begin{align}
  \log V_i &\sim \mathcal{N}\left(\mu_{V}, \sigma_V^2\right) \\
  \log \mathrm{ka}_i &\sim \mathcal{N}\left(\mu_{\mathrm{ka}}, \sigma_{\mathrm{ka}}^2\right) \\
  \log \mathrm{Cl}_i &\sim \mathcal{N}\left(\mu_{\mathrm{Cl}}, \sigma_{\mathrm{Cl}}^2\right) \\
  y_{ij}   &\sim \mathcal{N}\left(h\left(V_i, \mathrm{Cl}_i, \mathrm{ka}_i, t_{ij}\right), \sigma^2\right),
\end{align}
}%
where \(h\) is a first-order one-compartment model:
{%
\setlength{\belowdisplayskip}{1ex} \setlength{\belowdisplayshortskip}{1ex}
\setlength{\abovedisplayskip}{1ex} \setlength{\abovedisplayshortskip}{1ex} 
\[
   h\left(V_i, \mathrm{Cl}_i, \mathrm{ka}_i, t_{ij}\right)
   \triangleq
   \frac{d_i \mathrm{ka}_i }{ V_i \left(\mathrm{ka}_i - \mathrm{Cl}_i\right) } \left( \mathrm{e}^{-\frac{\mathrm{Cl}}{V_{i}} t_{ij}} - \mathrm{e}^{- \mathrm{ka}_i t_{ij}} \right),
\]
}%
and for each of the \(i\)th patient,
\begin{center}
  \vspace{-2ex}
  {\begingroup
  %\hspace{-1.em}
  \setlength\tabcolsep{1ex} 
  \begin{tabularx}{1.1\textwidth}{lX}
    \(y_{ij}\)        & is the concentration of the drug at \(t_{ij}\) (\textrm{mg/L}), \\
    \(t_{ij}\)        & is the time of the \(j\)th measurement (\textrm{hours}), \\
    \(d_i\)           & is the administered dosage (\textrm{mg}/\textrm{kg}), \\
    \(\mathrm{ka}_i\) & is the drug absorption rate, \\
    \(\mathrm{Cl}_i\) & is the drug's clearance, and \\
    \(V_i\)           & is the volume of the central compartment. \\
  \end{tabularx}
  \endgroup}
  \vspace{-1ex}
\end{center}
\(z_i = (\log V_i, \log \mathrm{ka}_i, \log \mathrm{Cl}_i) \in \mathbb{R}^{3}\) for \(i = 1, \ldots, n\) are the latent variables local to each patient sampled using MCMC, while the hyperparameters \(\theta = (\mu_{\mathrm{ka}}, \mu_{V}, \mu_{\mathrm{Cl}}, \sigma_{\mathrm{ka}}, \sigma_V, \sigma_{\mathrm{CL}}, \sigma) \in \mathbb{R}^3  \times \mathbb{R}_{>0}^4\) are inferred by maximizing the marginal likelihood.

 \begin{figure}
     \centering
     \includegraphics{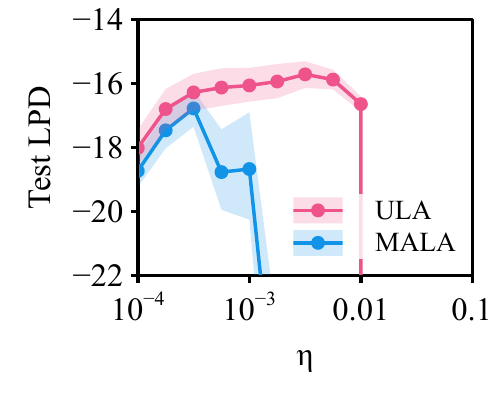}
     \vspace{-1ex}
     \caption{
     \textbf{Test average marginal log-predictive density (LPD) for the pharmacokinetics model versus the MALA/ULA stepsize \(\eta\).}
     The colored bands are 80\% bootstrap confidence intervals of the mean computed from \(32\) independent train-test splits of a ratio of \(9:3\).
     }
     \label{fig:pharma}
\vspace{-1ex}
 \end{figure}

\paragraph{Dataset}
We use the Theophylline dataset preprocessed and distributed by the \texttt{saemix} package~\citep{Comets2017parameter}, which is based on the one originally distributed by \texttt{NONMEM}~\citep{Boeckmann1994nonmem}.
This dataset contains 12 patients who were administered oral doses of the drug, where the concentration was measured 11 times over 25 hours.
(The \texttt{saemix} version excludes the measurement at \(t=0\), leaving 10 measurements per patient.)

We initialize at \(\theta_0 = (-1, 0, 0, 1, 1, 1, 1)\), run MCMC-SAEM for 1000 iterations after 100 initial burn-in MCMC steps.
We randomly split the patients into a train and test set of \(9:3\).
Then, we estimate the marginal log-predictive density of the test patients resulting from the hyper-parameters found by SAEM.
For estimating the test marginal LPD, we use the average of 100 importance weights drawn using annealed importance sampling (\citealp{Neal2001annealed}), each using 1000 annealing steps with a quadratic schedule.

\paragraph{Results}
The results are shown in \Cref{fig:pharma}.
We can see that ULA converges for the widest range of stepsizes.
In this example, MALA struggles the most because the likelihood is highly non-smooth.
Misspecifications of the hyper-parameters result in a sudden increase in the rejection rate.
Since ULA is immune to this problem, it makes constant progress as long as it doesn't diverge.

\begin{table}[t]
%\vspace{-ex}
\caption{POISSON GLM DATASETS}\label{table:count_data}
\vspace{-2ex}
\begin{center}
\begin{tabular}{lrrr}
\textbf{NAME}  & \# data & \(\operatorname{dim}(z)\) & \(\operatorname{dim}(\theta)\)\\
\hline
\textsf{medpar} & 1495 & 1495 & 6 \\
\textsf{azpro}  & 3589 & 3589 & 4 \\
\end{tabular}
\end{center}
\vspace{-3ex}
\end{table}

\vspace{-1ex}
\subsection{Robust Poisson Regression}
\vspace{-1ex}
Our first realistic experiment is a generalized linear model (GLM) with a Poisson likelihood.
In particular, we consider a robustified, or ``localized'' \citep{Wang2018general}, Poisson regression model, also known as the Poisson-log-normal model~\citep[\S 4.2.4]{Cameron2013regression}.
The model is described as follows:
{%
\setlength{\belowdisplayskip}{1.5ex} \setlength{\belowdisplayshortskip}{1.5ex}
\setlength{\abovedisplayskip}{1.5ex} \setlength{\abovedisplayshortskip}{1.5ex} 
\begin{align}
   \eta_i &\sim \mathcal{N}\left( \beta^{\top}x_i + \beta_0, \sigma\right) \\
    y_i   &\sim \mathsf{Poisson}\left( \exp\left(\eta_i\right) \right).
\end{align}
}%
The hyper-parameters are \(\theta = (\beta, \beta_0, \sigma) \in \mathbb{R}^{d} \times \mathbb{R} \times \Rset_{>0} \).
Unlike the popular negative-binomial regression model, this model is non-conjugate and intractable.
We will apply MCMC-SAEM to perform maximum-likelihood inference of the regression coefficients while marginalizing the local response \(\eta_i\).
We run MCMC-SAEM for 100 iterations after 10 initial burn-in steps.
The considered datasets are shown in~\Cref{table:count_data} and were obtained from the \texttt{COUNT} package in R~\citep{Hilbe2016count}.

 \begin{figure}
     \centering
     \subfloat[\textsf{azpro}]{
       \includegraphics{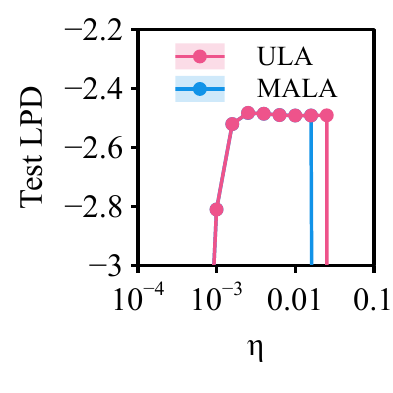}
       \vspace{-2ex}
     }
     \subfloat[\textsf{medpar}]{
       \includegraphics{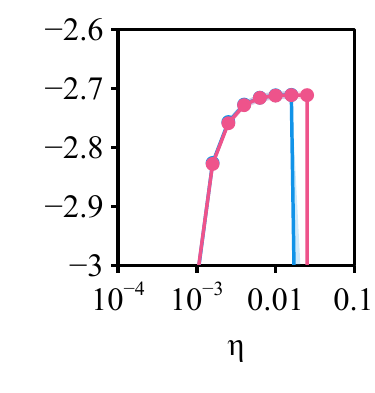}
       \vspace{-2ex}
     }
     \vspace{-1ex}
     \caption{
     \textbf{Test average log-predictive density (LPD) for robust Poisson regression versus the MALA/ULA stepsize \(\eta\).}
     ULA is more robust against the choice of stepsize on \textsf{azpro}.
     The colored bands are 80\% bootstrap confidence intervals of the mean computed from \(32\) independent train-test splits of a ratio of \(8:1\).
     }
     \label{fig:poisson}
\vspace{-1ex}
 \end{figure}

\begin{figure*}
\centering
\subfloat[\textsf{german}]{
  \includegraphics[scale=0.95]{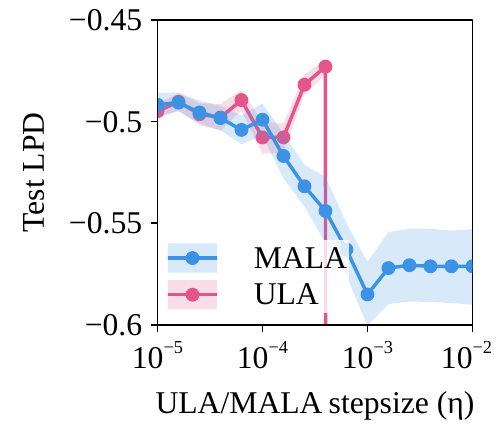}
  \vspace{-2.5ex}
}
\subfloat[\textsf{phishing}]{
  \includegraphics[scale=0.95]{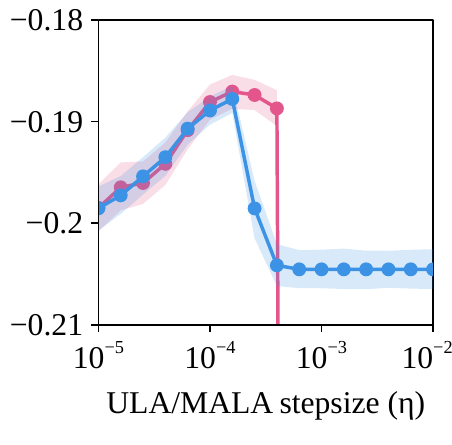}
  \vspace{-2.5ex}
} 
\subfloat[\textsf{caravan}]{
  \includegraphics[scale=0.95]{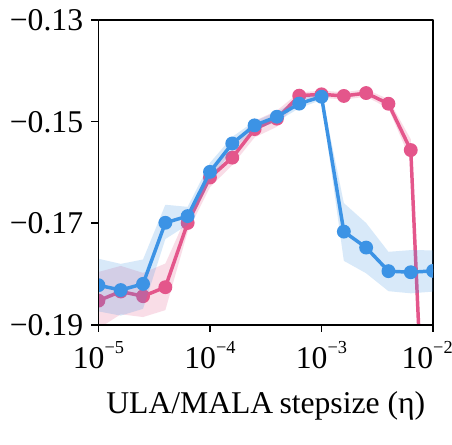}
  \vspace{-2.5ex}
} 
\caption{
  \textbf{Test average log-predictive density (LPD) for logistic regression with automatic relevance determination versus MALA/ULA stepsize \(\eta\).
  }
  The colored bands are 80\% bootstrap confidence intervals of the mean computed from \(32\) independent train-test splits of a ratio of \(8:1\).
}\label{fig:logisticard}
\vspace{-2ex}
\end{figure*}

\vspace{-1ex}
\paragraph{Results}
The results are shown in~\Cref{fig:poisson}.
%Since Poisson GLMs are regular in an optimization perspective, we can see that both ULA and MALA perform similarly.
As expected, ULA converges for a wider range of stepsizes.
However, in this example, the difference between the two methods is very small.
This is because both methods mix quickly for this posterior; it is strongly log-concave and factorizes into univariate posteriors.

\subsection{Logistic Regression with Automatic Relevance Determination}
Automatic relevance determination (ARD; \citealp{Neal1998bayesian,Mackay1996bayesian}) is prior on regression coefficients, where each regressor \(\beta_i\) is assigned its own scale parameter \(\gamma_i\).
When maximizing the marginal likelihood with respect to the relevance parameters \(\gamma_i\), the ARD prior is known to have a sparsifying effect, where irrelevant features are pruned as \(\gamma_i \to \infty\).
This ``shrinkage'' effect is lost if one does fully Bayesian inference. Therefore the empirical Bayes version of the problem is especially relevant.

Unfortunately, solving the maximum marginal likelihood problem is challenging,
even for linear regression models~\citep{Tipping2001sparse,Wipf2007new}.
Here, we demonstrate that MCMC-SAEM can be used to solve the ARD problem for logistic regression with Bernoulli likelihoods.

The model is described as:
{%
\setlength{\belowdisplayskip}{1ex} \setlength{\belowdisplayshortskip}{1ex}
\setlength{\abovedisplayskip}{1ex} \setlength{\abovedisplayshortskip}{1ex} 
\begin{align}
   \beta_0 &\sim \mathcal{N}\left(0, 10\right) \quad\beta \sim \mathcal{N}\left(0, \mathbf{\gamma}^{-1}\right) \\
   p_i &= \operatorname{logistic}\left({\beta^{\top}}{x_i} + \beta_0\right) \\
   y_i  &\sim \mathsf{Bernoulli}\left(p_i\right),
\end{align}
}%
where \(\mathbf{\gamma} = (\gamma_1, \gamma_2, \ldots, \gamma_d)\in \Rset_{>0}^{d}\).

\begin{table}[t]
\caption{LOGISTIC REGRESSION DATASETS}\label{table:logistic_data}
\vspace{-2ex}
\begin{center}
\begin{tabular}{lrrr}
\textbf{NAME}  & \# data & \(\operatorname{dim}(z)\) & \(\operatorname{dim}(\theta)\)\\
\hline
\textsf{phishing} & 11054 & 68 & 69 \\
\textsf{german}   & 1000 & 217 & 216 \\
\textsf{caravan}  & 9822 & 620 & 619 \\
\end{tabular}
\end{center}
\vspace{-3ex}
\end{table}

We optimize the hyperparameters \(\theta = \gamma\) using MCMC-SAEM for 2000 iterations, after 100 burn-in steps, starting from an initial point of \(\theta_0 = \left(1, \ldots, 1\right)\).
On this problem, the result was quite sensitive to the initial point.
After running MCMC-SAEM, we test the quality of the hyperparameters by estimating the log-predictive density (LPD) on a held-out test dataset using samples from the posterior.
The samples are separately drawn using 2000 steps of MALA after 2000 adaptation/burn-in steps.
MALA is automatically tuned to target an acceptance rate of \(0.57\) \citep{Roberts2001optimal} using Nesterov's dual averaging procedure~\citep{Nesterov2009primaldual}.
We replicate this over 32 independent train-test splits.
The datasets were obtained from the UCI repository~\citep{Dua2017uci} and are shown in \Cref{table:logistic_data}.
The categorical variables were one-hot encoded, while the continuous features were z-standardized.

Furthermore, for this problem, preconditioning MALA and ULA is crucial since the scale of the posterior greatly varies depending on whether a feature is pruned or not.
We use a diagonal preconditioner \(P\) where the diagonal is set as \(P_{ii} = \nicefrac{1}{\left(\gamma_i^2 + 0.01\right)} + \delta\) for \(\beta\) and \(1\) for \(\beta_0\).
\(\delta = 2 \times 10^{-16} > 0\) is necessary to ensure that the MCMC chain is not reducible even when a feature is pruned by \(\gamma_i \to \infty\).

\vspace{-1ex}
\paragraph{Results}
The results are shown in \Cref{fig:logisticard}.
We can see that ULA converges to a high quality solution for the widest range of stepsizes.
On \textsf{german}, only ULA achieves the highest level of accuracy.
Notably, in the large \(\eta\) regime, the MALA chain tended to reduce to a state with an acceptance rate close to 0.
ULA, on the other hand, is immune to this issue since it always makes progress as long as it does not diverge.
Furthermore, when \(\eta\) was too large, ULA immediately diverged at the initial SAEM iterations, which is easier to diagnose and correct.

\section{DISCUSSIONS}
In this work, we theoretically and empirically studied the impact of approximate MCMC algorithms.
The theory suggests that they are feasible for SAEM in high dimensions as soon as the marginal log-likelihood is smooth enough.
That is, the asymptotic bias will have a minimal effect on the found solution.
We empirically confirmed this fact on multiple statistical problems.
Furthermore, in our experiments, we observe that ULA versus MALA represents a trade-off between asymptotic bias versus non-asymptotic bias, where the latter can be more significant on high-dimensional and poorly conditioned problems. 
That is, with large stepsizes, MALA converges slower than ULA due to rejections, incurring a large non-asymptotic transient bias.
As a result, ULA converges faster on these problems.

On a different note, this work provides a clear use case of approximate MCMC algorithms in statistics.
While recent works~\citep{akyildiz2023interacting,kuntz2023particle,Valentinde2019efficient} in empirical Bayes estimation leveraging approximate MCMC didn't explore the \textit{benefits} of approximate MCMC methods over exact MCMC, we showed here that being approximate can, in fact, be better.

\newpage
\subsubsection*{Acknowledgements}
The authors would like to thank the anonymous reviewers for their critical and constructive feedback.

K. Kim was supported by a gift from AWS AI to Penn Engineering's ASSET Center for Trustworthy AI, J. R. Gardner was supported by NSF award [IIS-2145644].
Part of the work of A. O. Durmus is funded by the European Union (ERC, Ocean, 101071601).
Views and opinions expressed are, however, those of the author(s) only and do not necessarily reflect those of the European Union or the European Research Council Executive Agency.
Neither the European Union nor the granting authority can be held responsible for them.

\medskip
\bibliographystyle{plainnat}
\bibliography{bibliography}

\clearpage
\section*{Checklist}

 \begin{enumerate}
 \item For all models and algorithms presented, check if you include:
 \begin{enumerate}
   \item A clear description of the mathematical setting, assumptions, algorithm, and/or model.  \\
   Yes.
   \item An analysis of the properties and complexity (time, space, sample size) of any algorithm.  \\
   No. But we present a non-asymptotic convergence guarantee.
   
   \item (Optional) Anonymized source code, with specification of all dependencies, including external libraries. \\
   Yes.
 \end{enumerate}

 \item For any theoretical claim, check if you include:
 \begin{enumerate}
   \item Statements of the full set of assumptions of all theoretical results. \\
   Yes.
   \item Complete proofs of all theoretical results. [Yes/No/Not Applicable] \\
   Yes.
   \item Clear explanations of any assumptions. [Yes/No/Not Applicable] \\
   Yes.
 \end{enumerate}

 \item For all figures and tables that present empirical results, check if you include:
 \begin{enumerate}
   \item The code, data, and instructions needed to reproduce the main experimental results (either in the supplemental material or as a URL). \\
   Yes. The link to the repository is disclosed in \Cref{section:experiments}.
   
   \item All the training details (e.g., data splits, hyperparameters, how they were chosen). [Yes/No/Not Applicable] \\
   Yes.
   
   \item A clear definition of the specific measure or statistics and error bars (e.g., with respect to the random seed after running experiments multiple times). \\
    Yes. See the main text.
    
   \item A description of the computing infrastructure used. (e.g., type of GPUs, internal cluster, or cloud provider). \\
   Yes. In \Cref{section:resources}.
 \end{enumerate}

 \item If you are using existing assets (e.g., code, data, models) or curating/releasing new assets, check if you include:
 \begin{enumerate}
   \item Citations of the creator If your work uses existing assets. \\
   Yes.
   \item The license information of the assets, if applicable. \\
   Not applicable.
   \item New assets either in the supplemental material or as a URL, if applicable.  \\
   Not applicable.
   
   \item Information about consent from data providers/curators. \\
   Not applicable.
   
   \item Discussion of sensible content if applicable, e.g., personally identifiable information or offensive content. \\
   Not applicable.
 \end{enumerate}

 \item If you used crowdsourcing or conducted research with human subjects, check if you include:
 \begin{enumerate}
   \item The full text of instructions given to participants and screenshots.  \\
   Not applicable.
   
   \item Descriptions of potential participant risks, with links to Institutional Review Board (IRB) approvals if applicable. \\
   Not applicable.
   
   \item The estimated hourly wage paid to participants and the total amount spent on participant compensation.  \\
   Not applicable.
 \end{enumerate}
 \end{enumerate}

 \appendix

\onecolumn
{\hypersetup{linkcolor=black}
\tableofcontents
}

\newpage

\onecolumn
\newpage
\section{TECHNICAL ASSUMPTIONS ON THE MARKOV KERNEL}
\label{section:controling_the_fluctuations}
In this section, we give conditions that imply \Cref{hyp:noise_bias_control} in terms of a bound from below of the Markov kernel on a small set and a drift condition toward this small set (see \citep{nummelin1991poisson} for the definitions and main results). 
It offers some insights regarding what properties are essential for the state-dependent Markov kernels to have convergence guarantees. These conditions are verified for ULA \citep{Valentinde2019efficient}.

 Define, for \(\lyapD: \mcz \rightarrow[1, \infty)\) and \(g: \mcz \rightarrow \mathbb{R}^{d}\) the norm
  \[
  \|g\|_{\lyapD}=\sup _{x \in \mcz} \frac{|g(x)|}{\lyapD(x)}
  \]
  
 Denote, for \(\lyapD: \mcz \rightarrow[1, \infty), \mathcal{L}_{V}:=\left\{g: \mcz \rightarrow \mathbb{R}^{d}, \sup _{z \in \mcz}\|g\|_{\lyapD}<\infty\right\}\).
 \begin{assumptionnonsup}
\label{hyp:DRI}
  For any \(s\in \Rset^d, \Kerpi[s][\eta]\) is \(\psi\)-irreducible and aperiodic \footnote{We use in this article the standard terminology and the notations introduced in \citep[Chapter 4,5]{nummelin1991poisson}}. 
  In addition there exist a function \(\lyapD: Z \rightarrow\) \([1, \infty)\), a constant \(l_c \geq 2\)  such that for any compact subset \(\mathcal{K} \subset \Rset^d\),
 \begin{enumerate}
 
     \item (DRI1) there exist an integer \(m\), constants \(0<\lambda<1, b, \kappa, \delta>0\) and a probability measure \(\nu\) such that
 
 \[
 \begin{array}{ll}
 \sup _{s \in \mathcal{K}} (\Kerpi[s][\eta])^{m} \lyapD^{l_c}(z) \leq \lambda \lyapD^{l_c}(z)+b \mathds{1}_{\mathrm{C}}(z), & \\
 \sup _{s \in \mathcal{K}} \Kerpi[s][\eta] \lyapD^{l_c}(z) \leq \kappa \lyapD^{l_c}(z) & \forall z \in \mcz \\
 \inf _{s \in \mathcal{K}} (\Kerpi[s][\eta])^{m}(z, A) \geq \delta \nu(A) & \forall z \in \mathrm{C}, \quad \forall A \in \mathcal{B}(\mcz) .
 \end{array}
 \]
 
 \item (DRI2) $\left\|S(.)\right\|_{\lyapD} <\infty $.
 
 \item (DRI3) there exists \(C\) such that, for all $\left((s,\eta), (s^{\prime},\eta')\right) \in (\mathcal{K}\times ]0,\eta_0])^2$
 \[
 \begin{array}{ll}
 \left\|\Kerpi[s][\eta] g-\Kerpi[s'][\eta'] g\right\|_{\lyapD} \leq C\|g\|_{\lyapD}\left|(s,\eta)-(s^{\prime},\eta')\right| & \forall g \in \mathcal{L}_{\lyapD}, \\
 \left\|\Kerpi[s][\eta] g-\Kerpi[s'][\eta'] g\right\|_{\lyapD^{l_c}} \leq C\|g\|_{\lyapD^{l_c}}\left|(s,\eta)-(s^{\prime},\eta')\right|, & \forall g \in \mathcal{L}_{\lyapD^{l_c}}
 \end{array}
 \]
 \end{enumerate}
\end{assumptionnonsup}
  Assumption (DRI1) is classical in the Markov chain literature; it implies the existence of a stationary distribution \(\pi_{s,\eta}\) for all \(s,\eta \in \Rset^d\times ]0,\eta_0]\) and \(\lyapD^{l_c}\)-uniform ergodicity, i.e. for each \(s,\eta \in \Rset^d\times ]0,\eta_0]\) there exist constants \(C_{s,\eta}<\infty\) and \(\gamma_{s,\eta} \in[0,1)\), such that for any function \(f \in \mathcal{L}_{\lyapD^{l_c}}\) and any integer \(k>0\)
 \[
 \left\|(\Kerpi[s][\eta])^{k} f-\pi_{s,\eta}(f)\right\|_{\lyapD^{l_c}} \leq C_{s,\eta} \gamma_{s,\eta}^{k}\|f\|_{\lyapD^{l_c}} .
 \]
 
 Note that the constants \(C_{s,\eta}\) and \(\gamma_{s,\eta}\) may be bounded over the compact sets of \(\Rset^d\), i.e. for each \(\mathcal{K} \subset \Rset^d\), there exists \(\bar{C}<\infty\) and \(\bar{\gamma} \in[0,1)\), such that \(\sup _{s \in \mathcal{K}\times ]0,\eta_0]} C_{s,\eta} \leq \bar{C}\) and \(\sup _{s,\eta \in \mathcal{K}\times ]0,\eta]} \gamma_{s,\eta} \leq \bar{\gamma}\). 
 The regularity of the kernels \(s,\eta \rightarrow \Kerpi[s][\eta]\) expressed in \(\lyapD\) and \(\lyapD^{l_c}\) norm is naturally less classical in (DRI3). 
 These conditions can be weakened by considering subgeometric ergodicity conditions \citep{debavelaere2021convergence} .
 (DRI2) is just a control on the summary statistic depending on the drift of the Markov kernel.

By \Cref{hyp:DRI}, we can control the solution of the Poisson equation related to the Markov kernels $\{\Pi_s^\eta,\ee s,\eta\in \Rset^d\times ]0,\eta_0]\} $ which are helpful to control the markovian noise $(e_n)$, these controls are presented in the following assumption.
\begin{assumptionnonsup}
    \label{hyp:A3}
  For any \(s\in \Rset^d\), the Poisson equation \(\poinu[s][\eta]-\Kerpi[s][\eta] \poinu[s][\eta]=S(.)-\pi_{s,\eta}\left(S(.)\right)\) $\pi_{s,\eta}\left(S(.)\right)=\check{s}_\eta(s)$ has a solution \(\poinu[s][\eta]\). 
   There exist a function \(W: \mcz \rightarrow[1, \infty]\) such that \(\{x \in \mcz, W(x)<\infty\} \neq \emptyset\),
    constants \(\alpha \in(0,1], l_c \geq 2\) such that for any compact subset \(\mathcal{K} \subset \Rset^d\), by denoting $\mathcal{K}'=\mathcal{K}\times ]0,\eta_0]$ the following holds:
  \[
  \begin{aligned}
  & \left\|S(.)\right\|_{W}<\infty,  \sup _{(s,\eta)\in \mathcal{K}'}\left(\left\|\poinu[s][\eta]\right\|_{W}+\left\|\Kerpi[s][\eta] \poinu[s][\eta]\right\|_{W}\right)<\infty, \\
  & \sup _{\left((s,\eta), (s',\eta')\right) \in \mathcal{K}'}\left|\theta-\theta^{\prime}\right|^{-\alpha}\left\{\left\|\poinu[s][\eta]-\poinu[s'][\eta']\right\|_{W}+\left\|\Kerpi[s][\eta] \poinu[s][\eta]-\Kerpi[s'][\eta'] \poinu[s'][\eta']\right\|_{W}\right\}<\infty .
  \end{aligned}
  \]
\end{assumptionnonsup}
We stress here the fact that the function \(W\) is global but that the bounds in the previous equations depend on the particular compact \(\mathcal{K}\) under consideration.
\begin{lemma}
    \label{lemma:DRI_A3}
    Assume \Cref{hyp:DRI}. Then \Cref{hyp:reg_kernel}, \Cref{hyp:A3} are satisfied.
\end{lemma}

It is \citep[Proposition 6.1]{andrieu2005stability} adapted to our framework here $\Theta=\Rset^d\times ]0,\eta_0], \beta=1$. We removed the implications depending on the context of \citep{andrieu2005stability}.

Furthermore, if we control the learning steps $(\gamma_n)$, the bias paramter $(\eta_n)$ and momentums of the latent variables $(Z_n)$, the desired result can be established.
\begin{assumptionnonsup}
    \label{hyp:learning_step}
    The sequence $(\gamma_n)$ and $(\eta_n)$ are noninscreasing, positive and satisfy $\sum_{k=0}^{\infty} \gamma_k=\infty$, $\lim_{k\to \infty}\gamma_k=0$, $\limsup_{k\to \infty}|\gamma_k^{-1}-\gamma_{k+1}^{-1}|=0$ and 
    $$\sum_{k=1}^{\infty} \{\gamma_k^2 +\gamma_k|\eta_{k+1}-\eta_k|^\alpha  +\gamma_k^{1+\alpha} \}<\infty $$
    where $\alpha$ is defined in \Cref{hyp:A3}.
\end{assumptionnonsup}
\begin{assumptionnonsup}
    \label{hyp:momentumW}
    For any compact $\mathcal{K}\subset \Rset^d$, there exists a constant $C>0$ such that for any \(z \in \mcz\),
  \[
  \sup _{s \in \mathcal{K}}\sup_{k\geq 0^*} \mathbb{E}_{z, s}\left[W^{l_c}\left(Z_{k}\right) \mathds{1}_{\{\sigma(\mathcal{K}) \geq k\}}\right] \leq C W^{l_c}(z).
  \]
  $$\sigma(\mathcal{K})=\inf\{k\geq 0: s_k\notin \mathcal{K} \}\cup\{+\infty\}\eqsp, $$
  $W$ is defined in \Cref{hyp:A3}.
\end{assumptionnonsup}

\begin{theorem}
    \label{thm:simplification}
    Assume \Cref{hyp:DRI}, \Cref{hyp:learning_step} and \Cref{hyp:momentumW} then \Cref{hyp:noise_bias_control} is satisfied.
\end{theorem}

The proof is postponed to \Cref{proof:simplification}.

The assumption \Cref{hyp:momentumW} is essential and cannot be recovery directly from \Cref{hyp:DRI}, if we want this hypothesis, 
we have to change a little
the recursion (2) to ensure that $(s_n)$ stays in a compact by design and that $|s_{n+1}-s_n|$ is controlled by a non increasing sequence $(\epsilon_n)$ \citep[p.9]{andrieu2005stability}.
We did not choose this framework because it gives rise to technicalities that are rarely implemented in practice. Even if it offers a guarantee of convergence, there are some drawbacks regarding convergence speed.

\newpage
\section{PROOFS}
\subsection{Asymptotic Convergence}
In this section, we will prove \Cref{thm:asymptotic} in \Cref{section:asymptotic} and \Cref{thm:simplification} in \Cref{section:controling_the_fluctuations}.

\subsubsection{Auxiliary Lemma}
Before to prove the main theorem, we introduce another technical Lemma similar to Lemma 1,
\begin{lemma}
  Assume \Cref{assumption:exponential_family}-\labelcref{hyp:regularity_loss}, then for any $s\in \Rset^d$,
  \begin{equation}
    \label{eq:eq_V_h_mm}
    \nabla V(s)=-A(s)h(s)\quad \text{and} \quad |\nabla V(s)|^2/\lambda_M\leq \left\langle\nabla V(s)|h(s)\right\rangle \leq |\nabla V(s)|^2/\lambda_m \eqsp .
  \end{equation}
\end{lemma}
\begin{proof}
In the proof of \citep[Lemma 2]{delyon1999convergence},
the authors derive that for any $s\in \Rset^d$, $\nabla V(s)=-A(s)h(s)$.
 By \Cref{hyp:regularity_loss}, $A(s)$ is invertible and $A^{-1}(s)$ has repectively for minimal and maximal eigen values $1/\lambda_M$ and $1/\lambda_m$, thus, writing for any $s\in\Rset^d$, $h(s)=-A^{-1}(s)\nabla V(s) $ yields the inequality \eqref{eq:eq_V_h_mm}.
\end{proof}

\newpage
\subsubsection{Proof of \Cref{thm:asymptotic}}\label{section:proof_biaspropagation}
Except for some inequalities and constants, the proof is nearly the same as in \citep{tadic2017asymptotic}.
 The notations of \citep{tadic2017asymptotic} and ours are not the same: the function $V$ and the sequences $(s_n)$,$(\gamma_i),(\beta_n)$ in this paper are replaced by the function $f$ and the sequences $(\theta_n)$, $(\alpha_i),(\eta_n)$ in their paper.
  In the following, we detail all the proof modifications in \citep{tadic2017asymptotic} to get the desired theorem.
  We expose the sketch of proof related to each part of the proof before to details the changes.

  \paragraph{Proof of \ref{thm:asymp:item1}}
First remark that $\mss=\{h(s)=0: s\in \Rset^d \}$ by 

  In \citep[p.15]{tadic2017asymptotic}, the idea is to consider a perturbated flow equations related to the velocity field $-\nabla f$ with the perturbation size $\beta>0$, for any $t\geq 0$,
  \begin{equation}
    \label{eq:perturbated}
    \frac{\dd \theta(t) }{\dd t} \in F_\beta(\theta(t))^{\text{tadic}}, \quad F_\beta(\theta(t))^{\text{tadic}}=\{-\nabla \theta(s(t))+v, v\in \Rset^{d_\theta}: |v|\leq \beta\} \eqsp .
  \end{equation} 
  Then, by \citep[Proposition 4.1, Theorem 5.7]{benaim2006dynamics}, on the event $\Lambda_Q$, all limit points of $(\theta_n)$ are in a set $R_{Q,2\beta}$ called "recurrent set" which is related to \eqref{eq:perturbated}[$\beta$]: there is a link between the asymptotics of the discretized and the continuous flow.
  By \citep[Theorem 3.1]{benaim2012perturbations}, working with \eqref{eq:perturbated} instead of the discretization, there exists $\psi_Q$ : $[0, \infty) \rightarrow$ $[0, \infty)$ 
  (depending only on $f(\cdot)$ ) such that $\lim_{t \rightarrow 0} \psi_Q(t)=\psi_Q(0)=0$ and $ R_{Q,\beta} \subset \{x\in \Rset^d : d(x,R_{Q,0})\leq \psi_Q(\beta/2) \} $ for any $\beta\geq 0$.
  Finally, by \citep[Proposition 4]{hurley1995chain}, we can show that $ R_{Q,0}=\{ x\in\Rset^d: \nabla f(x)=0\} $ provided that $f$ is $p$-continuously differentiable with $p>d$ by using Sard's theorem.

  To adapt the reasoning, we take $h$ instead of $-\nabla f$ such that we consider $F_\beta(s(t))=\{h(s)+v, v\in \Rset^d: |v|\leq \gamma\}$ for any $t,\beta \geq0$.
  \citep[Theorem 3.1]{benaim2012perturbations} and \citep[Proposition 4.1, Theorem 5.7]{benaim2006dynamics} remain usable as long as $h$ is continuous, which is given by \Cref{hyp:regularity_loss}.
\citep[Proposition 4]{hurley1995chain} can't be applied directly. This Proposition is for gradient flow, i.e.
$\dd x/\dd t=-\nabla f(x) $ with $f$ continuously differentiable.
 However, the proof of \citep[Proposition 4]{hurley1995chain} can still be done by replacing $-\nabla f$ by $h$ and $f$ by $V$.
 Indeed, for any $s\in\Rset^d$ and $K\subset \Rset^d$ a compact such that $\phi([0,1],s)\subset K$,
 $$V(\phi(0,s))-V(\phi(1,s))=-\int_0^1 \frac{\dd V(\phi(t,s))}{\dd t} \dd t=-\int_0^1 \underbrace{\left\langle\nabla V(\phi(t,s))|h(\phi(t,s))\right\rangle}_{=F(\phi(t,s))} \dd t$$
 $$\geq \frac{1}{\lambda_M}\int_0^1  |\nabla V(\phi(t,s))|^2 \dd t \eqsp , $$
 where we use \eqref{eq:eq_V_h_mm}. This shows that the following flow $\dd x/\dd t=h(x) $ decreases an energy $V$.
  The others arguments in the proof of \citep[Proposition 4]{hurley1995chain} remains the same as long as we remark that the regular points of $V$ are dense by Sard's theorem and \Cref{hyp:regularity_loss}, and by \eqref{eq:eq_V_h_mm} we have
  \begin{equation}
    \mss=\{s\in\Rset^d:\nabla V(s)=0 \}=\{s\in\Rset^d:\nabla h(s)=0 \}\eqsp .
  \end{equation}
  The proof of \ref{thm:asymp:item1} is complete.
  
\paragraph*{Proof of \ref{thm:asymp:item2}-\ref{thm:asymp:item3}}
We work on the event $\Lambda_Q$, i.e., the sequence $(s_n)$ is in the compact $Q$ for $n$ large enough. 

In \citep[Proposition 8.2]{tadic2017asymptotic}, the authors show that for any $s\in \Rset^d$, the distances $d(s,\mss),d(f(s),\mss)$ can be bounded from above by a term of the form $a_1|\nabla f(s)|^{a_2}$ with $a_1,a_2>0$ using geometric and regularity arguments. 
The remaining work is to bound $|\nabla f(s)|$, which is performed in \citep[Proposition 8.3]{tadic2017asymptotic} using previous results \citep[Lemma 8.1-3]{tadic2017asymptotic} all deriving from a taylor expansion 
of $f(s_{a(n,t)})-f(s_n)$ in $s_n$ given in \citep[equation (32) p.17]{tadic2017asymptotic}.
 The idea is to measure the influence of the bias $\beta$ and the noise $e_n$ compared to the influence of the velocity field $\nabla f$: if $\nabla f$ has more influence than the other terms, $f$ will decrease, and its gradient as well till the point when $\nabla f$ has less influence than the noise and the bias.
 They show that the norm of the gradient at the asymptotic is ruled only by the bias since the accumulated noise vanishes by \Cref{hyp:noise_bias_control}.

 In order to adapt the proof, we should only change the Taylor expansion \citep[equation (32) p.17]{tadic2017asymptotic} by replacing $-\nabla f$ by $h$ and then propagating the changes in the analysis related to the norm of the gradient in \citep[Lemma 8.1-3]{tadic2017asymptotic} by using \eqref{eq:eq_V_h_mm}.

% To quantify the bias in the asymptotics, the idea is to measure the influence of the bias $\beta$ and the noise $e_n$ compared to the influence of the velocity field $-\nabla f$.
% The workhorse of this analysis is a taylor expansion of $f(s_{a(n,t)})-f(s_n)$ in $s_n$ given in \citep[equation (32) p.17]{tadic2017asymptotic}.
% Then, it is natural to quantify the volume of $A_{Q,\epsilon}=\{f(\theta):\theta\in Q,\eqsp |\nabla f(\theta) |\leq \epsilon  \} $ for any $\epsilon$, which is done using the Yomdin theorem \cite[Theorem 1.2]{yomdin1983geometry} in \citep[Proposition 8.1]{tadic2017asymptotic}.

 By doing a Taylor expansion, for any $n\geq 0$, $t\in (0,\infty)$,
$$ V(s_{a(n,t)})-V(s_n)=  \sum_{i=n}^{a(n,t)-1} \gamma_i \left\langle\nabla V(s_n)| h(s_n)\right\rangle-\left\langle\nabla V(s_n)|\sum_{i=n}^{a(n,t)-1} \gamma_i \xi_i\right\rangle+|\phi_n(t)|$$
where $\phi_n$ is defined in \citep[p.17]{tadic2017asymptotic}. Then, by using \eqref{eq:eq_V_h_mm},
$$ V(s_{a(n,t)})-V(s_n)\leq - |\nabla V(s_n)|\left(\frac{1}{\lambda_M}|\nabla V(s_n)|\sum_{i=n}^{a(n,t)-1} \gamma_i-|\sum_{i=n}^{a(n,t)-1} \gamma_i \xi_i|\right)+|\phi_n(t)|$$

We denote by $\phi=\limsup_{n\to \infty } |\nabla V(s_n)|$ and $C_{1,Q}=\sup_{s\in Q}|\nabla V(s)| $. The second change are in equations (38),(39) in \citep[Lemma 8.1, p.18]{tadic2017asymptotic} which are replaced for any $t\in (0,\infty)$ by 
$$\limsup_{n\to \infty} \max_{n\leq k<a(n,t)} | V(s_k)-V(s_n)|\leq C_{1,Q} t(\phi+\beta)\max\left(\frac{1}{\lambda_m},1\right)$$
$$\limsup_{n\to \infty} |\phi_n(t)|\leq C_{1,Q} t^2(\phi+\beta)^2\max\left(\frac{1}{\lambda_m},1\right)^2$$
since the last equation in \citep[p.18]{tadic2017asymptotic} has to be replaced by, 
$$|s_k-s_n|\leq \sum_{i=n}^{k-1}\gamma_i |h(s_i)|+|\sum_{i=n}^{k_1}\gamma_i \xi_i|\leq  t(\phi+\epsilon)\max\left(\frac{1}{\lambda_m},1\right)+\max_{n\leq i <a(n,t)} |\sum_{i=n}^{k_1}\gamma_i \xi_i|$$ 
where we used again \eqref{eq:eq_V_h_mm} and the fact that $|\nabla V(s_n)|\leq \phi $ for $n$ large enough as in \citep{tadic2017asymptotic}.

 From these changes, the constants used in the proof are modified.
The constant $\gamma$ is replaced p.20 by $\gamma= 2\lambda_M(\epsilon +\beta) $ and $C_{2,Q}=4\lambda_M M_Q $.
The \citep[Lemma 8.3]{tadic2017asymptotic} holds on $(\Lambda_Q\setminus N_0)\cap (\phi>2\beta \lambda_M )$.
It impacts the final constant, we have $K_Q= \lambda_M \max\left(2,\tilde{C}_Q,C_{2,Q}\right)$ instead of $\max\left(2,\tilde{C}_Q,C_{2,Q}\right)$, where $M_Q$ is a constant related to the Yomdin theorem \citep[Proposition 8.1, 8.2]{tadic2017asymptotic} and $\tilde{C}_Q = \sup_Q |\nabla V(s)|$.

\newpage
\subsubsection{Proof of \Cref{thm:simplification}}
\label{proof:simplification}

We follow the same scheme of proof of \citep[Proposition 5.2]{andrieu2005stability}.
We will intensely use the conditions \Cref{hyp:A3} given by \Cref{lemma:DRI_A3} and \Cref{hyp:DRI}.
$\mathbb{E}_{s, z}$ denotes the conditionnal expectation given $s_0=s$ and $Z_0=z$.
Let $\mathcal{K}\subset \Rset^d$ be a compact set and set the event $\tilde{\Lambda}_\mathcal{K}=\cap_{k=0}^\infty (s_k\in \mathcal{K})$.
Using the existence of solution to the Poisson equation,
for any $k\geq 1$, we have,
$$e_{k-1}=S(z_k)-\check{s}_{\eta_k}(s_{k-1})=(I-\Kerpi[s_{k-1}][\eta_k])\nu_{s_{k-1}}^{\eta_k}(z_k)\eqsp ,$$
and we denote by,
$$T_n=\sum_{k=1}^n \gamma_k e_{k-1} \mathds{1}_{\sigma(\mathcal{K},\epsilon)\geq k}\eqsp . $$
We want to show that $T_n$ converges almost surely on the event $\tilde{\Lambda}_\mathcal{K}$. Using for all $k\geq 0$:
$$\mathds{1}_{\sigma(\mathcal{K})\geq k}=\mathds{1}_{\sigma(\mathcal{K}) \geq k+1}+\mathds{1}_{\sigma(\mathcal{K}) = k}\eqsp, $$
me may write $T_n= \sum_{i=1}^5 T_n^{(i)}$ where,
\[
    \begin{aligned}
    \text{(A.1)}&\ee T_n^{(1)}=\sum_{k=1}^n \gamma_k\left(\poinu[s_{k-1}][\eta_k]\left(Z_k\right)-\Kerpi[s_{k-1}][\eta_k]\poinu[s_{k-1}][\eta_k]\left(Z_{k-1}\right)\right) \mathds{1}_{\{\sigma(\mathcal{K})\geq k\}}\eqsp, \\
    \text{(A.2)}& \ee T_n^{(2)}=\sum_{k=1}^{n-1} \gamma_{k+1}\left(\Kerpi[s_{k}][\eta_{k+1}] \poinu[s_{k}][\eta_{k+1}]\left(Z_k\right)-\Kerpi[s_{k-1}][\eta_k]\poinu[s_{k-1}][\eta_k]\left(Z_k\right)\right) \mathds{1}_{\{\sigma(\mathcal{K})  \geq k+1\}}\eqsp, \\
    \text{(A.3)}& \ee T_n^{(3)}=\sum_{k=1}^{n-1}\left(\gamma_{k+1}-\gamma_k\right) \Kerpi[s_{k-1}][\eta_k]\poinu[s_{k-1}][\eta_k]\left(Z_k\right) \mathds{1}_{\{\sigma(\mathcal{K}) \geq k+1\}} \eqsp ,\\
    \text{(A.4)}& \ee T_n^{(4)}=\gamma_1 \Kerpi[s_{0}][\eta_1] \poinu[s_{0}][\eta_1]\left(Z_0\right) \mathds{1}_{\{\sigma(\mathcal{K}) \geq 1\}}-\gamma_n \Kerpi[s_{n-1}][\eta_n]\poinu[s_{n-1}][\eta_n]\left(Z_n\right) \mathds{1}_{\{\sigma(\mathcal{K}) \geq n\}}\eqsp, \\
    \text{(A.5)}& \ee T_n^{(5)}=-\sum_{k=1}^{n-1} \gamma_k \Kerpi[s_{k-1}][\eta_k]\poinu[s_{k-1}][\eta_k]\left(Z_k\right) \mathds{1}_{\{\sigma(\mathcal{K})=k\}}
\end{aligned}
\]
We show now the convergence a.e of \(T_{n}^{(i)}, i=1, \ldots, 5\) on the event $\tilde{\Lambda}_\mathcal{K}$. First remark that $T_n^{(5)}=0$ on $\tilde{\Lambda}_\mathcal{K}$.
 In the sequel \(C\) denotes a constant which depends 
 only upon the compact set \(\mathcal{K}\) through the quantities defined 
 in the assumptions and whose value may change upon each appearance.
Denoting by 
  \[
      D( \boldsymbol{\gamma}, \mathcal{K}, z)=\sup _{s \in \mathcal{K}}  \sup _{k \geq 1} \mathbb{E}_{z, s}\left[W^{l_c}\left(Z_{k}\right) \mathds{1}_{\{\sigma(\mathcal{K}) \geq k\}}\right]\eqsp ,
      \]
      we have for any $s\in \Rset^d$ and $z\in \mcz$
  \[
    \begin{aligned}
        \text{(A.6)}\quad
    &  \mathbb{E}_{s, z}\left[\sum_{k=1}^\infty \gamma_k^2\left|\poinu[s_{k-1}][\eta_k]\left(Z_k\right)-\Kerpi[s_{k-1}][\eta_k]\poinu[s_{k-1}][\eta_k]\left(Z_{k-1}\right)\right|^2 \mathds{1}_{\{\sigma(\mathcal{K})\geq k\}}\right] \\
    &\leq C\left(\sum_{k=0}^{\infty} \gamma_k^{2}\right)^{l_c / 2} D( \boldsymbol{\gamma}, \mathcal{K}, z), \\
    \text{(A.7)}\quad
    & \mathbb{E}_{s, z}\left[\sum_{k=1}^{\infty} \gamma_{k+1}\left|\Kerpi[s_{k}][\eta_{k+1}] \poinu[s_{k}][\eta_{k+1}]\left(Z_k\right)-\Kerpi[s_{k-1}][\eta_k]\poinu[s_{k-1}][\eta_k]\left(Z_k\right)\right| \mathds{1}_{\{\sigma(\mathcal{K})  \geq k+1\}} \right] \leq  \\
    &C\left[\left(\sum_{k=1}^{\infty} |\gamma_k|^{1+\alpha}\right)^{\frac{l_c}{1+\alpha}}+\left(\sum_{k=1}^{\infty} \gamma_k |\eta_{k+1}-\eta_k|^{\alpha}\right)^{l_c} \right] D( \boldsymbol{\gamma}, \mathcal{K}, z)^{1+\alpha}, \\
    \text{(A.8)}\quad
    &  \mathbb{E}_{s, z}\left[\sum_{k=1}^{\infty}\left|\gamma_{k+1}-\gamma_k\right| |\Kerpi[s_{k-1}][\eta_k]\poinu[s_{k-1}][\eta_k]\left(Z_k\right) |\mathds{1}_{\{\sigma(\mathcal{K}) \geq k+1\}}\right] \\
    &\leq C \left(\sum_{k=1}^{\infty}|\gamma_{k}-\gamma_{k+1}|\right)^{l_c} D( \boldsymbol{\gamma}, \mathcal{K}, z) \\
    \text{(A.9)}\quad
    & 
     \mathbb{E}_{s, z}\left[ \gamma_n^{l_c} |\Kerpi[s_{n-1}][\eta_n]\poinu[s_{n-1}][\eta_n]\left(Z_n\right)|^{l_c} \mathds{1}_{\{\sigma(\mathcal{K}) \geq n\}}\right] \leq C \gamma_k^{l_c} D( \boldsymbol{\gamma}, \mathcal{K}, z) .
    \end{aligned}
    \]

We will show these inequalities at the end of the proof.
For any $k\geq 1$,
\[
\begin{aligned}
& \mathbb{E}_{z, s}\left[\left(\poinu[s_k][\eta_{k+1}]\left(Z_{k+1}\right)-\Kerpi[s_k][\eta_{k+1}] \poinu[s_k][\eta_{k+1}]\left(Z_{k}\right)\right) \mathds{1}_{\{\sigma(\mathcal{K}) \geq k+1\}} \mid \mathcal{F}_{k}\right]= \\
&\left(\Kerpi[s_k][\eta_{k+1}] \poinu[s_k][\eta_{k+1}]\left(Z_{k}\right)-\Kerpi[s_k][\eta_{k+1}] \poinu[s_k][\eta_{k+1}]\left(Z_{k}\right)\right) \mathds{1}_{\{\sigma(\mathcal{K}) \geq(k+1)\}}=0
\end{aligned}
\]
\(T_{n}^{(1)}\) is a \(\left(\mathbb{R}^{d}\right.\)-valued) martingale. Applying Doobs theorem with (A.6), we have the convergence of $(T_n^{(1)})$ almost surely on $\tilde{\Lambda}_\mathcal{K}$.
Since \(T_{n}^{(5)} \mathds{1}_{\{\sigma(\mathcal{K}) \geq n\}}=0\), we have

\[
T_{n} \mathds{1}_{\{\sigma(\mathcal{K}) \geq n\}}=\sum_{i=1}^{4} T_{n}^{(i)} \mathds{1}_{\{\sigma(\mathcal{K}) \geq n\}} .
\]
Using the inequalities (A.7)-(A.9) with \Cref{hyp:gamma_n} gives the convergence of $T_{n}^{(i)}$ for $i=2,3,4$ on the event $\tilde{\Lambda}_\mathcal{K}$. On the event 
$\{\sup_n |s_n|<\infty \} $, we can work in $\tilde{\Lambda}_\mathcal{K}$ for some compact set $\mathcal{K}\subset \Rset^d$, then $\sum_{i}\gamma_{i+1} e_i$ converges almost surely, which yields the assumption \Cref{hyp:noise_bias_control} is given.

\paragraph{Proof of (A.6)}
    
Under \Cref{hyp:A3}, we have for any $k\geq 1$,
$$ 
\left|\poinu[s_{k-1}][\eta_k]\left(Z_k\right)-\Kerpi[s_{k-1}][\eta_k]\poinu[s_{k-1}][\eta_k]\left(Z_{k-1}\right) \right| \mathds{1}_{\{\sigma(\mathcal{K}) \geq k+1\}} 
\leq
W(Z_k) \mathds{1}_{\{\sigma(\mathcal{K}) \geq k+1\}}+W(Z_{k-1}) \mathds{1}_{\{\sigma(\mathcal{K}) \geq k\}}
$$
and then,
\begin{align}
\mathbb{E}_{s, z}\left[\sum_{k=1}^\infty \gamma_k^2\left|\poinu[s_{k-1}][\eta_k]\left(Z_k\right)-\Kerpi[s_{k-1}][\eta_k]\poinu[s_{k-1}][\eta_k]\left(Z_{k-1}\right)\right|^2 \mathds{1}_{\{\sigma(\mathcal{K})\geq k\}}\right]
&\leq 
\mathbb{E}_{s, z}\left[\sum_{k=1}^\infty 4 \gamma_k^2 W(Z_{k})^2  \mathds{1}_{\{\sigma(\mathcal{K})\geq k\}}\right] 
\\
&\leq C(\sum_{k=1}^\infty  \gamma_k^2)^{\frac{l_c}{2}} D( \boldsymbol{\gamma}, \mathcal{K}, z)\eqsp ,
\end{align}
where we used that $\{\gamma_i\}$ is decreasing and a Minkowski inequality.
 
\paragraph{Proof of (A.7)}
 Under \Cref{hyp:A3}, we have
    %First we write, 
    %$$ D_k=\Kerpi[s_k][\eta_{k+1}] \poinu[s_k][\eta_{k+1}]\left(Z_{k}\right)-\Kerpi[s_{k-1}][\eta_{k}] \poinu[s_{k-1}][\eta_{k}]\left(Z_{k}\right)=$$
    %$$\underbrace{\Kerpi[s_k][\eta_{k+1}] \poinu[s_k][\eta_{k+1}]\left(Z_{k}\right)-\Kerpi[s_k][\eta_{k}] \poinu[s_k][\eta_{k}]\left(Z_{k}\right)}_{=D_k^1}+\underbrace{\Kerpi[s_k][\eta_{k}] \poinu[s_k][\eta_{k}]\left(Z_{k}\right)-\Kerpi[s_{k-1}][\eta_{k}] \poinu[s_{k-1}][\eta_{k}]\left(Z_{k}\right)}_{=D_k^2} $$
    %First we can bound $D_k^2$,
For any $k\geq 1$,
$$
\begin{aligned}
\gamma_{k+1}\left| \Kerpi[s_k][\eta_{k+1}] \poinu[s_k][\eta_{k+1}]\left(Z_{k}\right)-\Kerpi[s_{k-1}][\eta_{k}] \poinu[s_{k-1}][\eta_{k}]\left(Z_{k}\right)\right | \mathds{1}_{\{\sigma(\mathcal{K}) \geq k+1\}}  
\leq C \gamma_{k+1} W\left(Z_{k}\right)\left|(s_{k},\eta_{k+1})-(s_{k-1},\eta_{k})\right|^{\alpha} \mathds{1}_{\{\sigma(\mathcal{K}) \geq k+1\}}
\end{aligned}
$$
then, using the bound on sufficient statistics by \Cref{hyp:A3} and that $s_k$ is in a compact,  for any $k\geq 1$,
$$
\begin{aligned}
    |s_{k}-s_{k-1}| \mathds{1}_{\{\sigma(\mathcal{K})  \geq k+1\}}&\leq \gamma_k\left|S(Z_{k+1})-s_k \right| \mathds{1}_{\{\sigma(\mathcal{K})  \geq k+1\}}\\
    &\leq \gamma_k (W(Z_k)+C)\mathds{1}_{\{\sigma(\mathcal{K}) \geq k+1\}}\eqsp ,
\end{aligned}
$$
Thus, we have,
$$
\begin{aligned}
    & \sum_{k=1}^{\infty} \gamma_{k+1}\left|\Kerpi[s_{k}][\eta_{k+1}] \poinu[s_{k}][\eta_{k+1}]\left(Z_k\right)-\Kerpi[s_{k-1}][\eta_k]\poinu[s_{k-1}][\eta_k]\left(Z_k\right)\right| \mathds{1}_{\{\sigma(\mathcal{K})  \geq k+1\}}    \\
    & \leq C\sum_{k=0}^{\infty} \gamma_{k+1} \left(|\eta_{k+1}-\eta_k|^{\alpha}+\gamma_k^{\alpha} (W(Z_k)+C)^{\alpha}\right) W\left(Z_{k}\right) \mathds{1}_{\{\sigma(\mathcal{K}) \geq k+1\}} \\
    &\leq C \sum_{k=0}^{\infty} \gamma_{k+1} |\eta_{k+1}-\eta_k|^{\alpha}W\left(Z_{k}\right)+\gamma_k^{1+\alpha}(W(Z_k)+ W(Z_k)^{\alpha+1})  \mathds{1}_{\{\sigma(\mathcal{K}) \geq k+1\}}\eqsp ,
\end{aligned}
$$
in the last inequality, we used that $\{\beta_k \}$ is bounded and that $(\gamma_k)$ is decreasing. We conclude using the Minkowski's inequality.

Remark that if $\beta>0$, we have $\epsilon_k=|s_{k}-s_{k-1}|=\Theta(\gamma_k)$ which makes it impossible to apply the proof of \citep{delyon1999convergence} to show that we can have $(s_n)$ in a compact by design. Indeed, it is needed to have $\sum_n (\frac{\gamma_n}{\epsilon_n})^{l_c}<\infty$, but here we have $\epsilon_n\sim \gamma_n$ because of the bias.

\paragraph{Proof of (A.8)}
Under \Cref{hyp:A3},
$$
\begin{aligned}
\sum_{k=1}^{n-1}|\gamma_{k+1}-\gamma_{k}| \left|\Kerpi[s_{k-1}][\eta_{k}] \poinu[s_{k-1}][\eta_{k}]\left(Z_{k}\right)\right| \mathds{1}_{\{\sigma(\mathcal{K}) \geq k+1\}} 
\leq C \sum_{k=1}^{\infty}\left|\gamma_{k}-\gamma_{k+1}\right| W\left(Z_{k}\right) \mathds{1}_{\{\sigma(\mathcal{K}) \geq k+1\}},
\end{aligned}
$$

and the proof follows from Minkowski's inequality.

\paragraph{Proof of (A.9)}
Under \Cref{hyp:A3},
$$
\mathbb{E}_{s, z}\left[\left|\Kerpi[s_{n-1}][\eta_{n-1}] \poinu[s_{n-1}][\eta_{n-1}]\left(Z_{n}\right) \mathds{1}_{\{\sigma(\mathcal{K}) \geq n\}}\right|^{l_c} \right]\leq C  \mathbb{E}_{s, z}\left[W^{l_c}\left(Z_{n}\right) \mathds{1}_{\{\sigma(\mathcal{K}) \geq n\}} \right] 
$$

\newpage
\subsection{Non-Asymptotic Convergence}\label{section:proof_nonasymptotic}
In this section, we will prove \Cref{thm:nonasymptotic} in \Cref{section:Non_asymptotic}.

\subsubsection{Auxiliary Lemmas}
Before proving the theorem, we need some preliminary results. 
First, we derive a Robbins-Siegmund Lemma.
 Then, we control the Markov stochasticity of the process $(e_n)$ by using the Poisson solutions given by \Cref{assumption:poisson_equation}.
  It brings out a weighted sum of Martingale difference, which can be bounded by using a concentration inequality given in \citep[Lemma 1]{Li2020Ahigh}. Finally, we conclude by reorganizing the terms.
  We first state the Robbins-Siegmund Lemma before giving others technical results.
\begin{lemma}{(Robbins-Siegmund Lemma)}\label{thm:descentlemma}
  Assume \Cref{assumption:exponential_family}-\labelcref{hyp:regularity_loss}, \Cref{assumption:bounded_update}, \Cref{assumption:asymptotic_bias}, and \Cref{assumption:smoothness}.
  Then,
  \begin{equation}
    V\left(s_{k+1}\right)\leq V\left(s_{k}\right) +
    \gamma_{k+1}\left\langle\nabla V\left(s_k\right)\middle|h\left(s_k\right)\right\rangle 
    + \gamma_{k+1} \left\langle \nabla V\left(s_k\right) \middle| e_k \right\rangle 
    + a_k + b_k \abs{h\left(s_k\right)}^2\eqsp ,
   \end{equation}
  where the constants are
  \begin{align}
      &a_k = \frac{1}{2} \lambda_M \sqrt{\tau_0}\gamma_{k+1} + L_V \sigma^2 \gamma_{k+1}^2 \eqsp, \\
      &b_k = \left(\frac{1}{2} \lambda_M \sqrt{\tau_0}+ \lambda_M \sqrt{\tau_1}\right) \gamma_{k+1} + L_V \gamma_{k+1}^2\eqsp .
  \end{align}
\end{lemma}
The proof is delayed after the small technical results that followed.

Then, we state the concentration inequality.
\begin{lemma}[Adaptation of Lemma 1 by \citealt{Li2020Ahigh}]\label{lemma:martingale}
     Let $(E_{n})_{n\geq 1}$ be a martingale difference sequence adapted to the filtration $\left(\mathcal{F}_{n+1}=\sigma(Z_i,i\leq n+1)\right)_{n\geq 0}$ and a stochastic process $(v_n)_{n\geq 0}$ adapted to the filtration $(\mathcal{F}_{n})_{n\geq 0}$, such that for any $k\geq 0 $,
       \(\mathbb{E}^{\mathcal{F}_{k+1}}\exp\left(E_k^2/v_k^2\right) \leq \exp(1)\) .
     Then, for any fixed \(\varrho > 0\) and \(\delta \in (0, 1)\), with probability at least \(1 - \delta\), it holds that for any $n\geq 1$,
     \[n
       \sum^{n}_{k=1} E_k \leq \frac{3}{4} \varrho \sum^{n}_{k=1} v_k^2 + \frac{1}{\varrho} \log \frac{1}{\delta}.
     \]
 \end{lemma}
 This lemma will be used with the following:
  \begin{lemma}
    \label{lemma:little_technical}
  Assume \Cref{assumption:exponential_family}-\labelcref{hyp:regularity_loss} and \Cref{assumption:poisson_equation}.
  Then, there exists $c\in ]0,2L_{\nu}^{(0)}]$ such that the MCMC kernel has sub-Gaussian tails for \(\nu_{s}^{\eta}\), i.e., for any $s\in\Rset^d$,
 \begin{equation}
  \label{eq:assumption_subgaussian}
   \mathbb{E}^{\mathcal{F}_{k+1}}\exp\left( {\left| \nu_{s}^{\eta}(Z_{k+1}) - \Pi_{s}^{\eta} \nu_{s}^{\eta}(Z_k) \right|}^2 / c^2 \right)
   \leq
   \exp\left(1\right).
 \end{equation}
 \end{lemma}
 The proof is straightforward:
 by \Cref{assumption:poisson_equation}, for any $k\geq 0 $ and $s\in \Rset^d$,
 $$\left| \nu_{s}^{\eta}(Z_{k+1}) - \Pi_{s}^{\eta} \nu_{s}^{\eta}(Z_k) \right| / 2L_{\nu}^{(0)}  \leq 1 \eqsp ,$$
 which implies that \eqref{eq:assumption_subgaussian} holds with at least $c=2L_{\nu}^{(0)}$.

 We introduce the Martingale difference related to \Cref{lemma:martingale}.
  By using the solutions of the Poisson equation by $\mathbf{N}2$, we have for any $k\geq 0$, 
 \[
   D_k 
   \triangleq 
   \nu_{s_k}^{\eta}\left(Z_{k+1}\right) - \Pi_{s_k}^{\eta} \nu_{s_k}^{\eta}\left(Z_{k}\right),\eqsp \mathbb{E}(D_k|\mathcal{F}_k)=0
 \] 
 where use the Markov property related to $(Z_k)$ to remark that $(D_k)_k$ is a sequence of Martingale difference adapted to the filtration $(\mathcal{F}_{k+1})_k$.
  This yields the following Lemma.
 \begin{lemma}
  \label{lemma:second_technical}
 Assume \Cref{assumption:exponential_family}-\labelcref{hyp:regularity_loss} and \Cref{assumption:poisson_equation}. Then, with probability at least \(1 - \delta\), for any $n\geq 1$, %the following Martingale difference sequence is bounded as
 \begin{equation}
   -\sum_{k=1}^{n} \gamma_{k+1} \left\langle \nabla V\left(s_k\right) \middle| D_k \right\rangle
   \leq 
   \varrho \lambda_M^2 c^2 \sum^{n}_{k=1} 
   \gamma_{k+1}^2  \abs{ h\left(s_k\right) }^2  + \frac{1}{\varrho} \log \frac{1}{\delta}\eqsp ,
 \end{equation}
 where \(\varrho > 0\) is a free variable.
 \end{lemma}

 \newpage
 \subsubsection{Proof of \Cref{lemma:second_technical}.}
 For any $k\geq 1$, denoting by
 \begin{alignat}{2}
   E_k \triangleq
   -\gamma_{k+1} \left\langle \nabla V\left(s_k\right) \middle| D_k \right\rangle \eqsp ,
 \end{alignat}
 we have,
 \begin{alignat}{2}
     \abs{E_k}^2
     &=
     \gamma_{k+1}^2 \abs{\left\langle \nabla V\left(s_k\right) \middle|D_k \right\rangle}^2
     \\
     &\leq
     \gamma_{k+1}^2 \abs{ \nabla V\left(s_k\right) }^2 \abs{ D_k }^2
     &&\qquad\text{(Cauchy-Schwarz)}
     \\
     &\leq
     \lambda_M^2 \gamma_{k+1}^2 \abs{ h\left(s_k\right) }^2 \abs{ D_k }^2 \eqsp .
     &&\qquad\text{(Lemma 1)}
 \end{alignat}
 Since \(D_k\) is sub-Gaussian by \Cref{lemma:little_technical}, $(E_k)$ is also sub-gaussian such that
 \[
   \mathbb{E}^{\mathcal{F}_{k+1}} \exp\left(\abs{E_k}^2/ \left(\lambda_M^2 \gamma_{k+1}^2 \abs{ h\left(s_k\right) }^2 c^2 \right) \right) \leq \exp\left(1\right)\eqsp .
 \]
 Thus, \Cref{lemma:martingale} applies after setting $v_k\triangleq \lambda_M^2 \gamma_{k+1}^2 \abs{ h\left(s_k\right) }^2 c^2 $, which is adapted to the filtration $\mathcal{F}_{k} $.
 Then, with probability at least \(1 - \delta\), for any $n \geq 1$,
 \begin{align}
   \sum_{k=1}^n E_k
   &\leq 
   \frac{3}{4} \varrho\sum^{n}_{k=1} v_k^2 + \frac{1}{\varrho} \log \frac{1}{\delta}
   \\
   &=
   \frac{3}{4} \varrho\sum^{n}_{k=1} 
   \lambda_M^2 \gamma_{k+1}^2 \abs{ h\left(s_k\right) }^2 c^2 + \frac{1}{\varrho} \log \frac{1}{\delta}
   \\
   &\leq
   \varrho \lambda_M^2 c^2 \sum^{n}_{k=1} 
   \gamma_{k+1}^2  \abs{ h\left(s_k\right) }^2  + \frac{1}{\varrho} \log \frac{1}{\delta}\eqsp ,
 \end{align}
 where we have only organized the constants in the last inequality.

\newpage
\subsubsection{Proof of \Cref{thm:descentlemma}}
  We now prove the Robbins-Siegmund Lemma.
  From the \(L_V\)-smoothness of the Lyapunov function by \Cref{assumption:smoothness}, we have for any $h\geq 0$,
  \begin{align}
    &V\left(s_{k+1}\right) \leq
    V\left(s_{k}\right)
    +
    \gamma_{k+1} 
    \left\langle
    \nabla V\left(s_k\right)
    \middle|
    H\left(s_k, Z_{k+1}\right)
    \right\rangle
    +
    \frac{L_V \gamma_{k+1}^2}{2} 
    {\left| H\left(s_k, Z_{k+1}\right) \right|}^2
    \\
    &=
    V\left(s_{k}\right)
    +
    \gamma_{k+1} 
    \left\langle
    \nabla V\left(s_k\right)
    \middle|
    h\left(s_k\right) + \xi_k
    \right\rangle
    +
    \frac{L_V \gamma_{k+1}^2}{2} 
    {\left| h\left(s_k\right) + H\left(s_k, Z_{k+1}\right) - h\left(s_k\right) \right|}^2\eqsp ,
  \shortintertext{applying the inequality \( {\left| a + b \right|}^2 \leq 2 {\left|a\right|}^2 + 2 {\left|b\right|}^2,\)}
    &\leq
    V\left(s_{k}\right)
    +
    \gamma_{k+1} 
    \left\langle
    \nabla V\left(s_k\right)
    \middle|
    h\left(s_k\right)+\xi_k
    \right\rangle
    +
    L_V \gamma_{k+1}^2
    \left(
    {\left| h\left(s_k\right) \right|}^2
    + {\left| H\left(s_k, Z_{k+1}\right) - h\left(s_k\right) \right|}^2
    \right)\eqsp ,
  \shortintertext{and applying \Cref{assumption:bounded_update},}
    &\leq
    V\left(s_{k}\right)
    +
    \gamma_{k+1} 
    \left\langle
    \nabla V\left(s_k\right)
    \middle|
    h\left(s_k\right)
    \right\rangle
    +
    \gamma_{k+1} 
    \left\langle
    \nabla V\left(s_k\right)
    \middle|
    \xi_k
    \right\rangle
    +
    L_V \gamma_{k+1}^2
    {\left| h\left(s_k\right) \right|}^2
    +
    L_V \gamma_{k+1}^2 \sigma^2
    \\
    &=
    V\left(s_{k}\right)
    +
    \underbrace{
    \gamma_{k+1} 
    \left\langle
    \nabla V\left(s_k\right)
    \middle|
    h\left(s_k\right)
    \right\rangle
    }_{\text{Mean-field dynamics}}
    +
    \underbrace{
    \gamma_{k+1} 
    \left\langle
    \nabla V\left(s_k\right)
    \middle|
    \xi_k
    \right\rangle
    }_{\text{MCMC bias dynamics}}
    +
    L_V \gamma_{k+1}^2
    \left(
    {\left| h\left(s_k\right) \right|}^2
    +
    \sigma^2
    \right)\eqsp .
  \end{align}
  
  For the bias dynamics, recall that the bias can be decomposed as
  \[
    \xi_k = \underbrace{e_k}_{\text{non-asymptotic bias}} + \underbrace{\beta_k}_{\text{asymptotic bias}}\eqsp .
  \]
  Considering this, the bias dynamics can be decomposed as
  \begin{align}
    \left\langle
    \nabla V\left(s_k\right)
    \middle|
    \xi_k
    \right\rangle
    =
    \left\langle
    \nabla V\left(s_k\right)
    \middle|
    e_k + \beta_k
    \right\rangle
    =
    \left\langle
    \nabla V\left(s_k\right)
    \middle|
    e_k
    \right\rangle
    +
    \left\langle
    \nabla V\left(s_k\right)
    \middle|
    \beta_k
    \right\rangle\eqsp .
  \end{align}
  For the asymptotic bias, we have
  \begin{align}
    \left\langle
    \nabla V\left(s_k\right)
    \middle|
    \beta_k
    \right\rangle
    &\leq
    \abs{
      \nabla V\left(s_k\right)
    }
    \abs{
      \beta_k 
    },
  \shortintertext{applying \Cref{lemma:preliminary} and \Cref{assumption:asymptotic_bias},} 
    &\leq
    \lambda_M
    \abs{
    h\left(s_k\right)
    }
    \sqrt{
      \tau_0+ \tau_1\abs{h\left(s_k\right)}
    }\eqsp ,
  \shortintertext{applying the inequality \( \sqrt{a + b} \leq \sqrt{a} + \sqrt{b}\), }
    &\leq
    \lambda_M
    \left(
    \sqrt{\tau_0} \abs{h\left(s_k\right)}
    + 
    \sqrt{\tau_1} \abs{h\left(s_k\right)}^2
    \right)\eqsp ,
  \shortintertext{and the inequality \(a \leq \nicefrac{1}{2}\left(1 + a^2\right)\) for \(a \geq 0\),}
    &\leq
    \lambda_M
    \left\{
    \frac{1}{2} \sqrt{\tau_0}
    +
    \left(\frac{1}{2} \sqrt{\tau_0} + \sqrt{\tau_1}\right) \abs{h\left(s_k\right)}^2
    \right\}\eqsp .
  \end{align}
  Combining the results, we have
  \begin{align}
    V\left(s_{k+1}\right)
    &\leq
    V\left(s_{k}\right)
    +
    \gamma_{k+1} 
    \left\langle
    \nabla V\left(s_k\right)
    \middle|
    h\left(s_k\right)
    \right\rangle
    +
    L_V \gamma_{k+1}^2
    \left(
    {\left| h\left(s_k\right) \right|}^2
    +
    \sigma^2
    \right)
    \\
    &\quad+
    \gamma_{k+1} 
    \left\langle
    \nabla V\left(s_k\right)
    \middle|
    e_k
    \right\rangle
    +
    \lambda_M \gamma_{k+1}
    \left(
    \frac{1}{2} \sqrt{\tau_0}
    +
    \left(\frac{1}{2} \sqrt{\tau_0} + \sqrt{\tau_1}\right) \abs{h\left(s_k\right)}^2
    \right)
    \\
    &=
    V\left(s_{k}\right)
    +
    \gamma_{k+1} 
    \left\langle
    \nabla V\left(s_k\right)
    \middle|
    h\left(s_k\right)
    \right\rangle
    +
    \gamma_{k+1} 
    \left\langle
    \nabla V\left(s_k\right)
    \middle|
    e_k
    \right\rangle
    \\
    &\quad
    +
    \left( \frac{1}{2} \lambda_M \sqrt{\tau_0}\gamma_{k+1} + L_V \sigma^2 \gamma_{k+1}^2 \right)
    +
    \left( \frac{1}{2} \lambda_M \sqrt{\tau_0}\gamma_{k+1} + \lambda_M \sqrt{\tau_1}\gamma_{k+1} + L_V \gamma_{k+1}^2\right) \abs{h\left(s_k\right)}^2\eqsp .
  \end{align}

\newpage
\subsubsection{Proof of \Cref{thm:nonasymptotic}}

    From \Cref{thm:descentlemma}, we have for any $k\geq 0$,
    \[
      -\gamma_{k+1} 
      \left\langle
      \nabla V\left(s_k\right)
      \middle|
      h\left(s_k\right)
      \right\rangle
      \leq
      V\left(s_{k}\right)
      -
      V\left(s_{k+1}\right)
      +
      \gamma_{k+1} 
      \left\langle
      \nabla V\left(s_k\right)
      \middle|
      e_k
      \right\rangle
      +
      a_k
      +
      b_k \abs{h\left(s_k\right)}^2\eqsp .
    \]
    Applying \Cref{lemma:preliminary},
    \[
      \gamma_{k+1} \lambda_m \abs{h\left(s_k\right)}_2^2
      \leq
      V\left(s_{k}\right)
      -
      V\left(s_{k+1}\right)
      +
      \gamma_{k+1} 
      \left\langle
      \nabla V\left(s_k\right)
      \middle|
      e_k
      \right\rangle
      +
      a_k
      +
      b_k \abs{h\left(s_k\right)}^2\eqsp .
    \]
    Summing this bound for \(k = 0, \ldots, n\) forms a telescoping sum as
    \begin{align}
      &\lambda_m 
      \sum^{n}_{k=0}
      \gamma_{k+1}
      \abs{h\left(s_k\right)}_2^2
      \\
      &\leq
      V\left(s_{0}\right)
      -
      V\left(s_{n+1}\right)
      +
      \sum^{n}_{k=0}
      \gamma_{k+1} 
      \left\langle
      \nabla V\left(s_k\right)
      \middle|
      e_k
      \right\rangle
      +
      \sum_{k=0}^{n} a_k
      +
      \sum^{n}_{k=0} b_k \abs{h\left(s_k\right)}^2
      \\
      &\leq
      V\left(s_{0}\right)
      -
      V^*
      +
      \underbrace{
      \sum^{n}_{k=0}
      \gamma_{k+1} 
      \left\langle
      \nabla V\left(s_k\right)
      \middle|
      e_k
      \right\rangle
      }_{\text{non-asymptotic bias dynamics}}
      +
      \sum_{k=0}^{n} a_k
      +
      \sum^{n}_{k=0} b_k \abs{h\left(s_k\right)}^2\eqsp .
    \end{align}
    
    Now, we focus on the non-asymptotic bias dynamics term.
    Karimi \textit{et al.} \cite[Theorem 2]{karimi2019non-asymptotic} show that the sum of inner products can be decomposed as
    \begin{align}
      -\sum_{k=0}^n
      \gamma_{k+1} 
      \left\langle
      \nabla V\left(s_k\right)
      \middle|
      e_k 
      \right\rangle
      &=
      A_1 + A_2 + A_3 + A_4 + A_5 \eqsp ,
    \end{align}
    where the terms are 
    \begin{align}
      A_1 &= -\sum_{k=1}^{n} \gamma_{k+1} \left\langle \nabla V\left(s_k\right) \middle| \nu^{\eta}_{s_k}\left(Z_{k+1}\right) - \Pi_{s_k}^{\eta} \nu^{\eta}_{s_k}\left(Z_{k}\right) \right\rangle\eqsp ,
      \\
      A_2 &= -\sum_{k=1}^{n} \gamma_{k+1} \left\langle \nabla V\left(s_k\right) \middle| \Pi_{s_k}^{\eta} \nu^{\eta}_{s_k}\left(Z_{k}\right) - \Pi_{s_{k-1}}^{\eta} \nu^{\eta}_{s_{k-1}}\left(Z_{k}\right) \right\rangle\eqsp ,
      \\
      A_3 &= -\sum_{k=1}^{n} \gamma_{k+1} \left\langle \nabla V\left(s_k\right) - \nabla V\left(s_{k-1}\right) \middle| \Pi_{s_{k-1}}^{\eta} \nu^{\eta}_{s_{k-1}}\left(Z_{k}\right) \right\rangle\eqsp ,
      \\
      A_4 &= -\sum_{k=1}^{n} \left( \gamma_{k+1} - \gamma_k \right) \left\langle \nabla V\left(s_{k-1}\right) \middle|  \Pi_{s_{k-1}}^{\eta} \nu^{\eta}_{s_{k-1}}\left(Z_{k}\right) \right\rangle\eqsp ,
      \\
      A_5 &= -\gamma_1 \left\langle \nabla V\left(s_{0}\right) \middle| \Pi_{s_{0}}^{\eta} \nu^{\eta}_{s_{k-1}}\left(z_{1}\right) \right\rangle
      +
      \gamma_{n-1} \left\langle \nabla V\left(s_{n-1}\right) \middle| 
      \Pi_{s_{n-1}}^{\eta} \nu^{\eta}_{s_{n-1}}\left(z_{n}\right)
       \right\rangle \eqsp .
    \end{align}
    \(A_1\) is given by \Cref{lemma:second_technical} where we set \(\varrho = 1\) and use that $c\leq2L_{\nu}^{(0)} $.
    On the other hand, for \(A_2, A_3, A_4, A_5\), we can use the results in the proof of \cite[Theorem 2]{karimi2019non-asymptotic} by setting
    \begin{align}
    c_0 = 0, \quad c_1 = \lambda_m, \quad d_0 = 0, \quad d_1 = \lambda_M, \\
    L = L_V, \quad L_{\mathrm{PH}}^{(0)} = L_{\nu}^{(0)}, \quad L_{\mathrm{PH}}^{(1)} = L_{\nu}^{(1)}.
    \end{align}
    Then, with probability at least \(1 - \delta\), we have
    \begin{align}
        A_1 
        &\leq 
         (2\lambda_M L_{\nu}^{(0)})^2 \sum^{n}_{k=1} 
        \gamma_{k+1}^2  \abs{ h\left(s_k\right) }^2  + \log \frac{1}{\delta} \eqsp ,
        \\
        A_2
        &\leq
        L_{\nu}^{(1)} \lambda_M 
        \left(
          \sigma \sum^{n}_{k=1} \gamma_k^2
          +
          \left( \frac{1}{2} + \alpha_1 \sigma + \alpha_1 \frac{1}{2} \right) \sum^{n}_{k=0} \gamma_k^2 \abs{h\left(s_k\right)}^2
        \right) \eqsp ,
        \\
        A_3
        &\leq
        L_V
        L_{\nu}^{(0)} 
        \left(
          \left(1+\sigma\right) \sum^{n}_{k=1} \gamma_k^2 + \sum^{n}_{k=1} \gamma_k^2 \abs{h\left(s_{k-1}\right)}^2
        \right)
        \\
        &\leq
        L_V
        L_{\nu}^{(0)} 
        \left(
          \left(1+\sigma\right) \sum^{n}_{k=1} \gamma_k^2 + \sum^{n}_{k=0} \gamma_{k+1}^2 \abs{h\left(s_{k}\right)}^2
        \right) \eqsp ,
        \\
        A_4
        &\leq
        L_{\nu}^{(0)} 
        \left(
          \left(\gamma_0 - \gamma_n\right)
          +
          \alpha_2 \lambda_M \sum^{n}_{k=1} \gamma_{k}^2 \abs{h\left(s_{k-1}\right)}^2
        \right)
        \\
        &\leq
        L_{\nu}^{(0)} 
        \left(
          \gamma_0
          +
          \alpha_2 \lambda_M \sum^{n}_{k=0} \gamma_{k+1}^2 \abs{h\left(s_{k}\right)}^2
        \right) \eqsp ,
        \\
        A_5
        &\leq
        L_{\nu}^{(0)} \lambda_M
        \left(
          2 + \sum^{n}_{k=0} \gamma_{k}^2 \abs{h\left(s_{k}\right)}^2
        \right)\eqsp .
    \end{align}
    By reorganizing the terms,
    \begin{align}
        \sum^{n}_{k=0} \gamma_{k+1} \left(\lambda_m - C_{b_1} - C_{n_1} \gamma_{k+1}\right) \abs{h\left(s_k\right)}^2
        \leq
        V\left(s_0\right) - V^*
        +
        C_{n_2} \sum^{n}_{k=0} \gamma_{k+1}^2
        +
        C_0 + \log\frac{1}{\delta} 
        +
        C_{b_2} \sum^{n}_{k=0} \gamma_{k+1}\eqsp .
    \end{align}
    From this, we obtain the condition on the stepsize that \(\gamma_{k+1} \leq \left( \lambda_m - C_{b_1} \right) / C_{n_1} \) for all \(k = 0, \ldots, n\).
    Furthermore, constant progress can be guaranteed by further enforcing \(\gamma_{k+1} \leq \frac{1}{2} \left( \lambda_m - C_{b_1} \right) / C_{n_1} \).
    Then, since \(\gamma_{k+1} \leq \gamma_k\), the following inequalities hold:
    \begin{align}
        \sum^{n}_{k=0} \gamma_{k+1} \abs{h\left(s_k\right)}^2
        &\leq
        \frac{2}{\lambda_m - C_{b_1}}
        \left(
        V\left(s_0\right) - V^*
        +
        C_{n_2} \sum^{n}_{k=0} \gamma_{k+1}^2
        +
        C_0 + \log\frac{1}{\delta} 
        +
        C_{b_2} \sum^{n}_{k=0} \gamma_{k+1}
        \right)
        \eqsp .
    \end{align}
    Finally, dividing both sides by \(\sum_{k=0}^n \gamma_{k+1}\), we have
    \begin{align}
        \frac{1}{\sum_{k=0}^n \gamma_{k+1}} \sum_{k=0}^n \gamma_{k+1} \abs{h\left(s_k\right)}^2
        \leq
        \frac{2}{\lambda_m - C_{b_1}}
        \left(
        \frac{
          V\left(s_0\right) - V^*
          +
          C_{n_2} \sum^{n}_{k=0} \gamma_{k+1}^2
          +
          C_0 + \log\frac{1}{\delta} 
        }{
          \sum_{k=0}^n \gamma_{k+1}
        }
        +C_{b_2}
        \right)
        \eqsp .
    \end{align}
    Since the lower bound forms a weighted average, lower bounding it with the minimum over the iterates yields the result.

\newpage
\section{COMPUTATIONAL RESOURCES}\label{section:resources}
\begin{table}[h]
  \centering
  \begin{threeparttable}
  \caption{Computational Resources}\label{table:resources}
  \begin{tabular}{ll}
    \toprule
    \multicolumn{1}{c}{\textbf{Type}}
    & \multicolumn{1}{c}{\textbf{Model and Specifications}}
    \\ \midrule
    System Topology & 1 socket with 8 physical cores \\
    Processor       & 1 Intel i9-11900F, 2.5 GHz (maximum 5.2 GHz) per socket \\
    Cache           & 80 KB L1, 512 KB L2, and 16 MB L3 \\
    Memory          & 64 GiB RAM \\
    \bottomrule
  \end{tabular}
  \end{threeparttable}
\end{table}
All experiments took about 50 hours to complete.

% \bibliographystyle{plainnat}
% %\bibliographystyle{plain}
% \bibliography{bibliography}

% \end{document}

\end{document}